\documentclass[11pt, reqno]{amsart}
\usepackage[dvipsnames]{xcolor}
\usepackage{amscd, amsmath, amssymb, amsthm, yfonts, mathrsfs, hhline, tikz}
\usetikzlibrary{arrows, decorations.markings, patterns}
\tikzset{>=latex}
\usepackage[centering,margin=1.2in,footskip=1cm]{geometry}
\usepackage[matrix,arrow,curve]{xy}

\usepackage{tikz-cd}

\global\binoppenalty=10000

\newtheorem{theorem}{Theorem}[section]

\newtheorem{prop}[theorem]{Proposition}
\newtheorem{cor}[theorem]{Corollary}

\theoremstyle{definition}
\newtheorem{defn}[theorem]{Definition}
\newtheorem{remark}[theorem]{Remark}
\newtheorem{notation}[theorem]{Notation}

\numberwithin{equation}{section}

\def\beq{\begin{equation}}
\def\eeq{\end{equation}}

\def\longra{\longrightarrow}

\newcommand{\hr}[1]{\left(#1\right)}                                                    % round, aligned
\newcommand{\hm}[1]{\left|#1\right|}                                                    % modulo, aligned
\newcommand{\ha}[1]{\left\langle#1\right\rangle}                                        % angle, aligned
\newcommand{\hs}[1]{\left[#1\right]}                                                    % square, aligned
\newcommand{\hc}[1]{\left\{#1\right\}}                                                  % calligraphic, aligned
\def\le{\leqslant}
\def\ge{\geqslant}

\def\bs{\boldsymbol}

\def\C{\mathbb C}

\def\eps{\varepsilon}

\def\Fc{\mathcal F}
\def\Frac{\operatorname{Frac}}

\def\gl{\mathfrak{gl}}
\def\Gc{\mathcal G}

\def\heis{\mathfrak{heis}}

\def\Hbb{\mathbb H}
\def\Hc{\mathcal H}
\def\Hilb{\mathbf{Hilb}}

\def\inv{\mathrm{inv}}

\def\la{\lambda}
\def\La{\Lambda}
\def\Lbb{\mathbb L}

\def\Mc{\mathcal{M}}
\def\Mtau{M^{\tau}}
\def\Mu{\mathrm{M}}

\def\R{\mathbb R}

\def\Re{\operatorname{Re}}
\def\Res{\mathrm{Res}}
\def\Rho{\mathrm P}

\def\Tc{\mathcal T}

\def\sh{\operatorname{sh}}

\def\SH{\mathbb{SH}}

\def\wt{\operatorname{wt}}
\def\Wc{\mathcal W}

\def\Xc{\mathcal X}

\def\Z{\mathbb Z}

\begin{document}

\title[Ruijsenaars wavefunctions as  modular group matrix coefficients]{Ruijsenaars wavefunctions as modular group matrix coefficients}
\author{Philippe Di Francesco}
\author{Rinat Kedem}
\author{Sergey Khoroshkin}
\author{Gus Schrader}
\author{Alexander Shapiro}
%\address{PDF: University of Illinois, IPhT Universit\'e Paris Saclay \\
%RK: University of Illinois \\
%SK: ??? \\
%GS: Northwestern University \\
%AS: University of Edinburgh.}

\begin{abstract}
We give a description of the Halln\"as--Ruijsenaars eigenfunctions of the 2-particle hyperbolic Ruijsenaars system as matrix coefficients for the order 4 element $S\in SL(2,\mathbb{Z})$ acting on the Hilbert space of $GL(2)$ quantum Teichm\"uller theory on the punctured torus. The $GL(2)$ Macdonald polynomials are then obtained as special values of the analytic continuation of these matrix coefficients. The main tool used in the proof is the cluster structure on the moduli space of framed $GL(2)$-local systems on the punctured torus, and an $SL(2,\mathbb{Z})$-equivariant embedding of the $GL(2)$ spherical DAHA into the quantized coordinate ring of the corresponding cluster Poisson variety. 

\end{abstract}

\maketitle

%\red{
%{\bf To do:}
%\begin{enumerate}
%\item References
%\item Reread the whole damn thing
%\end{enumerate}
%}

\section{Introduction}

The goal of this article is to show that the Halln\"as--Ruijsenaars eigenfunctions of the 2-particle relativistic hyperbolic Ruijsenaars system are the matrix coefficients for the element $S\in SL(2,\mathbb{Z})$ acting on the Hilbert space of $GL(2)$ quantum Teichm\"uller theory on the punctured torus.

Let us explain more precisely what we mean by this, and indicate why such a formula is both natural and useful.
The $N$-particle relativistic hyperbolic Ruijsenaars system~\cite{Rui87} is a system of $N$ commuting difference operators acting on meromorphic functions of $N$ complex variables. The analytic theory of this integrable system was studied extensively in~\cite{HR1, HR2, HR3} and~\cite{BDKK1, BDKK2, BDKK3, BDKK4}, where its joint eigenfunction transform was constructed and shown to define a unitary equivalence between appropriate Hilbert spaces.

The hyperbolic Ruijsenaars operators are close relatives of the Macdonald difference operators from the theory 
of symmetric functions. The action of the latter operators preserves the space of $S_N$-symmetric polynomial functions, where they act diagonalizably with distinct eigenvalues, and their joint eigenfunctions are the well-known Macdonald polynomials.

As discovered by Cherednik in~\cite{Che05}, it is useful to regard the Macdonald difference operators as generators of a commutative subalgebra in a 2-parameteric family of non-commutative algebras  $\Hbb_{q,t}(GL_N)$ known as the \emph{double affine Hecke algebra} (DAHA) associated to the group $GL_N$. Indeed, the DAHA is a quotient of the braid group of $N$ points on a 2-torus, and passing to it reveals a hidden $SL(2,\mathbb{Z})$-symmetry of Macdonald theory. 

In this article we will focus on the case $N=2$, in which connection to the geometry of the torus becomes even easier to describe: the fiber at $q=1$ of the \emph{spherical subalgebra} $\SH_{q,t}$ of the $GL_2$ DAHA recovers the coordinate ring of the $GL_2(\mathbb{C})$-character variety for the punctured torus $T^2\setminus D^2$, with the parameter $t$ related to the eigenvalues of the monodromy around the puncture. In this realization of the algebra, the Macdonald operators are quantizations of the observables given by the fundamental traces of the monodromy of a $GL(2)$-local system around the $(0,\pm1)$-curve on the torus. The commuting subalgebra in $\SH_{q,t}$ associated in the same way to the $(\pm1,0)$-curve is identified with that generated by the operators of multiplication by the elementary symmetric functions acting on the ring of $2$-variable symmetric polynomials. 

To reconnect this algebraic picture with the analytic Ruijsenaars theory, recall from the work of Fock and Goncharov that the punctured torus character variety (or more precisely, its framed version in which local systems are equipped with extra data of a Borel reduction near the puncture) is a \emph{cluster Poisson variety.} Hence by the results of~\cite{FG09}, its quantization delivers a (projective) unitary Hilbert space representation of the torus mapping class group $SL(2,\Z)$, which acts by explicit quantum cluster transformations. The $(0,\pm1)$-curve generators of the spherical DAHA $\SH_{q,t}$, which correspond to the Macdonald--Ruijsenaars operators in Cherednik's representation, act by unbounded, symmetric operators on this Hilbert space. On the other hand, the $(\pm1,0)$-curve generators act in the cluster representation by the Hamiltonians of the 2-particle $q$-difference \emph{open Toda chain.} 

In contrast to its Ruijsenaars counterpart, the quantum open relativistic Toda chain is generally regarded as one of the simplest nontrivial quantum integrable systems, and has a well-developed algebraic and analytic theory, see e.g. the following non-exhaustive list of references~\cite{Giv97, Eti99, Sev99, KLS02, GLO1, GLO2, GLO3, DFKT17, SS18}. In particular, it has a complete, orthogonal set of eigendistributions known as (non-compact, $b$-, or) $q$-Whittaker functions $\Psi_{\bs\la}$, that are symmetric in $\bs\la\in\mathbb{R}^2$. The main point of this article is that the unitary automorphism of the cluster representation given by the action of the element
$$
S=\begin{pmatrix}0&-1\\1&0\end{pmatrix}\in SL(2,\mathbb{Z})
$$
allows one to reduce questions about the Ruijsenaars system into much easier ones about the Toda system. In particular, in Theorem~\ref{theorem2} we obtain a description of the (suitably renormalized, {\it cf.} Proposition~\ref{prop:RH}) Halln\"as--Ruijsenaars eigenfunctions $\Phi_{\bs\mu}^\tau(\bs\la)$ as the matrix coefficient of the element $S\in SL(2,\mathbb{Z})$ between Whittaker functions:
\begin{align}
\label{eq:intro-eq}
\Phi_{\bs\mu}^\tau(\bs\la) = \ha{\Psi_{\bs\la}, S\Psi_{\bs\mu}}.
\end{align}
This matrix coefficient realization of the Halln\"as--Ruijsenaars eigenfunctions makes immediate most of their important properties, including the `bispectral' and `Poincar\'e duality' symmetries $\Phi_{\bs\la}^{\tau}(\bs\mu)=\Phi_{\bs\mu}^\tau(\bs\la)$ and $\Phi_{\bs\mu}^{-\tau}(\bs\la)=\Phi_{\bs\la}^\tau(\bs\mu)$, see Corollaries~\ref{cor:Macdo} and~\ref{cor:cor}.

From the physical point of view, the formula~\eqref{eq:intro-eq} expresses the Halln\"as--Ruijsenaars eigenfunctions as scalar products between eigenstates of two quantum mechanical integrable systems given by traces of holonomies along the $(1,0)$ and $(0,1)$-curves on the torus respectively. In other words, the Halln\"as--Ruijsenaars eigenfunctions are the transition matrix between two geometrically natural bases in the Hilbert space for quantum Teichmuller theory on $T^2\setminus D^2$. This transition matrix realizes the action of the $S$-move generator in the corresponding projective representation of the Moore--Seiberg groupoid, and has been previously computed by Teschner and Vartanov, see formula (6.30) of \cite{TV15}.

In this sense, formula~\eqref{eq:intro-eq} is quite analogous to the well-known description of the $SU_2$ Racah--Wigner $q$-$6j$ symbols which arises in an almost identical way, except that the punctured torus is replaced by the 4-punctured sphere. The semiclassical asymptotics of such scalar products were described geometrically in~\cite{Res18}: the two classical integrable systems determine two presentations of the character variety as Lagrangian fibrations, and the asymptotic expansion of the scalar product takes the form of a sum over intersection points of the two corresponding Lagrangian fibers.

We conclude by explaining in Theorem~\ref{thm:analytic-continuation} the sense in which the Macdonald polynomials and Harish-Chandra series can be recovered as special values of the analytically continued Halln\"as--Ruijsenaars eigenfunctions.

\subsection*{Acknowledgements}

We are grateful to Mingyuan Hu, Nicolai Reshetikhin, J\"org Teschner and Eric Zaslow for helpful discussions. P.D.F. is supported by the Morris and Gertrude Fine Endowment and the Simons Foundation Grant MP-TSM-00002262. R.K. acknowledges support from the Simons Foundation Grant MP-TSM-00001941. S.Kh. is grateful to the University of Edinburgh for its hospitality during the fall semester of 2023.  The work of S. Kh. was supported by the International Laboratory of Cluster Geometry of National Research University Higher School of Economics, Russian Federation Government grant, ag. No. 075-15-2021-608, dated 08.06.2021. G.S. has been supported by the NSF Standard Grant DMS-2302624. A.S. has been supported by the European Research Council under the European Union’s Horizon 2020 research and innovation programme under grant agreement No 948885 and by the Royal Society University Research Fellowship.

\section{Ruijsenaars and open Toda integrable systems}
In this section, we recall the definition of the two integrable systems that play a central role in the paper: the two-particle hyperbolic Ruijsenaars system, and the open $q$-difference Toda chain.

\subsection{Ruijsenaars system}
\label{sec:macdo-ops}
The $\gl_2$ Macdonald difference operators
$$
M_j(\bs\la; g \,|\,\bs\omega),\qquad j=1,2
$$
with periods $\omega_1, \omega_2$ and the coupling constant $g$ are the commuting operators\footnote{Note that our notations differ from those in~\cite{BDKK2}. Namely, $M_1$ and $M_2$ in \emph{loc. \!\!cit.\!} are respectively $M_1 M_2^{-1}$ and $M_2^{-1}$ here.}
\begin{align*}
M_1(\bs\la; g \,|\,\bs\omega) &= \frac{\sh\frac{\pi}{\omega_2}\left(\la_1-\la_2+i g\right)}{\sh\frac{\pi }{\omega_2}\left(\la_1-\la_2\right)} e^{i\omega_1\frac{\partial}{\partial \la_1}} + \frac{\sh\frac{\pi}{\omega_2}\left(\la_1-\la_2+i g\right)}{\sh\frac{\pi}{\omega_2}\left(\la_2-\la_1\right)} e^{i\omega_1\frac{\partial}{\partial \la_2}}, \\
M_2(\bs\la; g \,|\,\bs\omega) &= e^{i\omega_1\left(\frac{\partial}{\partial \la_1}+\frac{\partial}{\partial \la_2}\right)},
\end{align*}
where we write $\bs x$ for a vector $(x_1,x_2) \in \R^2$. Given $b\in\mathbb{R}$, let us specialize
\beq
\label{omega}
\omega_1 = b, \qquad \omega_2=b^{-1}
\eeq
and introduce the pure imaginary constant
$$
c_b = \frac{i(b + b^{-1})}{2}.
$$
We will work in the regime in which the coupling constant $g$ admits a parametrization
\beq
\label{g}
g = -i(c_b+2\tau) \qquad\text{for}\qquad \tau \in \R.
\eeq
Then for $j=1,2$ we recover Macdonald operators $\Mtau_j = M_j(\bs\la; -i(c_b+2\tau) \,|\,b, b^{-1})$, which take the form
\begin{align*}
\Mtau_1 &= \frac{t\La_1 - t^{-1}\La_2}{\La_1-\La_2}T_{\La_1} + \frac{t\La_2 - t^{-1}\La_1}{\La_2-\La_1}T_{\La_2}, \\
\Mtau_2 &= T_{\La_1}T_{\La_2},
\end{align*}
with
$$
\La_j = e^{2\pi b\la_j},\qquad T_{\La_j} = e^{ib\frac{\partial}{\partial \la_j}},
\qquad\text{and}\qquad t = e^{\pi b (2\tau+c_b)}.
$$
The operators $T_{\La_j}$ act on functions $f(\bs\La)$ as shift operators:
$$
T_{\La_1} f(\La_1,\La_2) = f(q^2\La_1,\La_2), \qquad T_{\La_2} f(\La_1,\La_2) = f(\La_1,q^2\La_2)
$$
where
$$
q = e^{\pi i b^2}.
$$
The Macdonald difference operators $\Mtau_j$ are essentially self-adjoint with respect to the Hermitian inner product $\ha{\cdot,\cdot}_S$  on  $L^2_{\mathrm{sym}}(\R^2, m(\bs\la)d\bs\la)$ defined by
\beq
\label{eq:Sklyanin-prod}
\ha{f,g}_S=\int_{\R^2} \overline{f(\bs\la)} g(\bs\la)m(\bs\la)d\bs\la,
\eeq
where
\beq
\label{Sklyanin}
m(\bs\la)d\bs\la =  2\sh(b(\lambda_1-\lambda_2))\sh(b^{-1}(\lambda_1-\lambda_2))d\bs\la
\eeq
is the $\gl_2$ Sklyanin measure.

Joint eigenfunctions of Macdonald operators were obtained in~\cite{HR1} and further studied in~\cite{HR2, HR3}, as well as~\cite{BDKK1}--\cite{BDKK4}. In the specialization~\eqref{omega},~\eqref{g}, these eigenfunctions take form
\beq
\label{eq:HR}
\Rho^\tau_{\bs\mu}(\bs\la) = \varphi(c_b-2\tau_-) e^{-\pi i \underline{\bs\la} \underline{\bs\mu}}
\int_\R e^{-4\pi i x(\la+\tau_-)} \frac{\varphi(x+\mu+\tau_-) \varphi(x-\mu+\tau_-)}{\varphi(x+\mu-\tau_-) \varphi(x-\mu-\tau_-)} dx,
\eeq
where $\varphi(z)$ is the \emph{non-compact quantum dilogarithm} discussed in the appendix,
$$
\tau_- = \frac{c_b}{2}-\tau,
\qquad\text{and}\qquad
\nu = \frac{\nu_1-\nu_2}{2}.
$$
As was shown in~\cite{BDKK2}, $\Rho^\tau_{\bs\mu}(\bs\la)$ are symmetric in the coordinate variables $\bs\la$, symmetric in the spectral variables $\bs\mu$, and satisfy the bispectral duality
$$
\Rho^\tau_{\bs\mu}(\bs\la) = \Rho^{-\tau}_{\bs\la}(\bs\mu).
$$

The Macdonald difference operators also act on the space of $\C(q,t)$-valued Laurent polynomials in $\La$, invariant under the involution $\La \mapsto \La^{-1}$. Their eigenbasis in the latter space is given by the celebrated symmetric Macdonald polynomials, which can be recovered as the Halln\"as--Ruijsenaars wavefunctions restricted to a lattice, see Section 6 for more details.

\begin{remark}
\label{rmk:analyticity-I}
Consider the Macdonald--Ruijsenaars measure
$$
m_g(\bs\la) = S_2(i(\la_1-\la_2)) S_2(i(\la_2-\la_1)+g),
$$
where $S_2$ is the Barnes double sine function from equation~\eqref{eq:barnes}.
The corresponding bilinear pairing $\hr{\cdot,\cdot}_R$ and Hermitian pairing $\ha{\cdot,\cdot}_R$ are defined by
$$
\hr{f,g}_R = \int_{\R^2}f(\bs\la)g(-\bs\la)m_g(\bs\la) d\bs\la, \qquad
\ha{f,g}_R = \int_{\R^2}f(\bs\la)\overline{g(\bs\la)}m_g(\bs\la) d\bs\la.
$$

\begin{figure}[h]
\begin{tikzpicture}
\fill[gray!20] (0,-2) rectangle (2,2);
\draw[->] (-1.0,0) -- (2.7,0);
\draw[dashed, ->] (0,-2) -- (0,2);
\draw[dashed] (2,-2)--(2,2);
\draw[very thick] (1,-2)--(1,2);
\draw[very thick] (0,0)--(2,0);
\fill(2.0,0) circle(1.5pt);
\node[below] at (2.25,0.05)	
{\tiny $2\omega$};
\fill(1.0,0) circle(1.5pt);
\node[below] at (1.2,0)	
{\tiny $\omega$};
\fill(0,0) circle(1.5pt); \node[below] at (0.15,0.05)	
{$\scriptstyle 0$};
\fill(1,1.5) circle(1.5pt); \node[below] at (1.45,1.55)	
{$\scriptstyle \omega+i\tau$};
%\node[below] at (5,-3) ;
%\node[below] at (5,1) {  $\scriptstyle \omega_1=b,\ \omega_2=b^{-1},\ \omega=\frac{b+b^{-1}}{2}$};
\end{tikzpicture}
\caption{Regions of unitarity of the Ruijsenaars system}
\label{fig:unitary-regime}   
\end{figure}
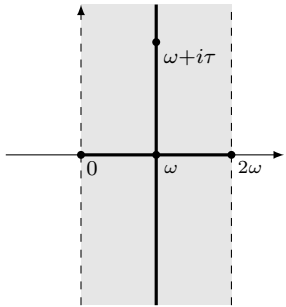 
As shown in e.g.~\cite{BDKK2}, for any periods $\omega_1, \omega_2$, and the coupling constant $g$ in the strip of analyticity of the Ruijsenaars system, that is $0 < \Re(g) < \omega_1+\omega_2$, the Macdonald operators are essentially self-adjoint with respect to the bilinear pairing $\hr{\cdot,\cdot}_R$. Moreover, as was shown in~\cite{BDKK3}, when the coupling constant $g$ is real, the Halln\"as--Ruijsenaars functions are orthogonal and complete with respect to the pairing $\ha{\cdot,\cdot}_R$. It is immediate from~\eqref{eq:qdl-unitary} and~\eqref{inv} that for $\tau \in i\R$ the Halln\"as--Ruijsenaars functions satisfy
$$
\Rho_{\bs\mu}^\tau(-\bs\la) = \zeta_{inv} e^{4\pi i \tau^2} \overline{\Rho_{\bs\mu}^\tau(\bs\la)},
$$
where the constant $\zeta_{inv}$ is given by~\eqref{eq:zetas}, and hence
$$
\hr{\Rho_{\bs\mu}^\tau,\Rho_{\bs\nu}^\tau}_R = \zeta_{inv} e^{4\pi i \tau^2} \ha{\Rho_{\bs\mu}^\tau,\Rho_{\bs\nu}^\tau}_R.
$$
Thus, we have two regimes of unitarity of the Ruijsenaars system, and Macdonald operators are symmetric in both of them. In the first regime the coupling constant $g$ is real and satisfies $0 < g < -ic_b$, equivalently $\tau$ is imaginary with $\hm{\tau} < c_b/2$, and the scalar product is $\ha{\cdot,\cdot}_R$. In the second regime $g = i(\tau -c_b)$ with an arbitrary real $\tau$, and the scalar product is $\ha{\cdot,\cdot}_S$, see~\eqref{eq:Sklyanin-prod}. In \cite{HR1}--\cite{HR3} and~\cite{BDKK1}--\cite{BDKK4} the Ruijsenaars system was studied in the first regime, in the present text we focus on the second. In Figure~\ref{fig:unitary-regime} we visualize the two regimes for $\omega_1=b$, $\omega_2=b^{-1}$, and $\omega=-ic_b$.
\end{remark}

\subsection{Open Toda chain}
The Hamiltonians of the $q$-difference $\gl_2$ open Toda system are the following commuting difference operators 
\begin{align*}
H_1 &= e^{2\pi b p_2} + e^{2\pi b(p_2+x_2-x_1)} + e^{2\pi b p_1}, \\
H_2 &= e^{2\pi b(p_1+p_2)}.
\end{align*}
These operators are evidently symmetric with respect to the Hermitian inner product $\ha{\cdot,\cdot}$ on $L^2(\R^2, d\bs x)$ given by
\beq
\ha{f,g}=\int_{\R^2} \overline{f(\bs x)} g(\bs x) d\bs x.
\eeq
They also commute with the \emph{$Q$-system discrete time evolution operator}~\cite{HKKR00, Wil15, DFK18, DFK24}, which is the unitary automorphism $\sigma_+$ of  $L^2(\R^2, d\bs x)$ given by
\beq
\label{eq:qsys}
\sigma_+ = e^{\pi i\hr{p_1^2+p_2^2}} \varphi(x_2-x_1).
\eeq
\begin{remark}
The reason for our choice of notation for the $Q$-system evolution operator is that it will be re-interpreted as a Dehn twist on the punctured torus in the topological reformulation given in Section~\ref{sec:cluster-torus}. See Remark~\ref{rmk:dehn} for details.
\end{remark}
For $\bs\la \in\mathbb{R}^2$, the \emph{Whittaker function} $\Psi_{\bs\la}(\bs x)$ is a common eigenfunction for the operators $H_1$, $H_2$, and $\sigma_+$ with the following eigenvalues, see~\cite{KLS02, SS18}:
\beq
\label{eigen}
\begin{aligned}
H_j \Psi_{\bs\la}(\bs x) &= e_j(\bs\la) \Psi_{\bs\la}(\bs x), \\
\sigma_+ \Psi_{\bs\la}(\bs x) &= \gamma \Psi_{\bs\la}(\bs x),
\end{aligned}
\eeq
where $e_j(\la)$ is the $j$-th elementary symmetric function in $e^{2\pi b \la_1}$, $e^{2\pi b \la_2}$, and
\beq
\label{eq:gamma-eigen}
\gamma = e^{\pi i\hr{\la_1^2+\la_2^2}}.
\eeq

The Whittaker function admits two explicit presentations, both of which utilize the following convention.

\begin{notation}
\label{contour-convention}
Throughout the paper, we will often consider contour integrals of the form 
$$
\int_{C} \prod_{j,k}\frac{\varphi(t-a_j)}{\varphi(t-b_k)}f(t)dt,
$$
where $f(t)$ is some entire function. Unless otherwise specified, the contour $C$ in such an integral is always chosen to be passing below the poles of $\varphi(t-a_j)$ for all $j$, above the poles of $\varphi(t-b_k)^{-1}$ for all $k$, and escaping to infinity in such a way that the integrand is rapidly decaying. 
\end{notation}

The Gauss--Givental representation of the Whittaker function reads
\beq
\label{GG}
\Psi_{\bs\la}(\bs x) = \zeta^{-1} e^{\pi i c_b (\la_2-\la_1)} e^{2\pi i \la_2 \underline{\bs x}} \int \frac{e^{2\pi i r(\la_1-\la_2) }} {\varphi(r-x_1-c_b)\varphi(x_2-r)}dr,
\eeq
where $\zeta$ is defined in~\eqref{eq:zetas}, and
$$
\underline{\bs x} = x_1+x_2.
$$
On the other hand, the function $\Psi_{\bs\la}(\bs x)$ can also be expressed as a Mellin--Barnes integral:
\beq
\label{MB}
\Psi_{\bs\la}(\bs x) = \zeta e^{\pi i \underline{\bs\la}(2x_2+c_b)} \int e^{2\pi i \alpha(x_1-x_2-c_b)} \varphi(\alpha-\la_1+c_b) \varphi(\alpha-\la_2+c_b) d\alpha.
\eeq
It is manifest from the Mellin--Barnes presentation that $\Psi_{\bs\la}(\bs x)$ is symmetric in the spectral variables $\bs\la$. The following theorem is the main result of~\cite{SS18}.
\begin{theorem}[SS18]
\label{thm:Whittaker}
Let $\Psi^{}_{\bs\la}(\bs x)$ be the {Whittaker} function for the $\gl_n$ $q$-difference open Toda chain. Define the \emph{Whittaker transform} by
\begin{align*}
\Wc \colon L^2(\R^n, d\bs x) &\longra L^2_{\mathrm{sym}}(\R^n, m(\bs\la)d\bs\la) \\
f(\bs x) &\longmapsto \int_{\R^n} f(\bs x) \overline{\Psi_{\bs\la}^{}(\bs x)} d\bs x.
\end{align*}
Then $\mathcal{W}$ defines a unitary equivalence between the Hilbert space of square-integrable functions in $\bs x$ and that of symmetric functions in $\bs\la$ that are square-integrable with respect to the $\gl_n$ Sklyanin measure $m(\bs\la)$ given by
\beq
\label{Sklyanin}
m(\bs\la)d\bs\la = \frac{2^{n(n-1)}}{n!} \prod_{\substack{j<k}}^n \sh(b(\la_j-\la_k))\sh(b^{-1}(\la_j-\la_k))d\bs\la.
\eeq
\end{theorem}
The Whittaker functions $\Psi_{\bs\la}(\bs x)$ also satisfy difference equations with respect to the spectral variables $\bs\la$.
Indeed, in terms of the \emph{dual Toda operators} $\check H_j$ defined by
\begin{align*}
\check H_1 &= \hr{t\Mtau_1}\Big|_{t=0} = \frac{1}{1-\La_1/\La_2}T_{\La_1} + \frac{1}{1-\La_2/\La_1}T_{\La_2}, \\
\check H_2 &= \Mtau_2,
\end{align*}
the bi-spectrality of the Whittaker functions can be expressed as follows: (see~\cite[Theorem 4.7]{SS19})
\begin{align*}
\Wc \circ H_1 &= e_1(\bs\la) \circ \Wc, & \Wc \circ e^{\pi b c_b} e^{2\pi b x_1} &= \check H_1 \circ \Wc, \\
\Wc \circ H_2 &= e_2(\bs\la) \circ \Wc, & \Wc \circ e^{2\pi b(x_1+x_2)} &= \check H_2 \circ \Wc.
\end{align*}
%Let us define the $n$-th twist
Following~\cite{DFK18}, let us define the $n$-th discrete time translate
$\check H_{j,n}$ of $\check H_{j,0}=\check H_j$ by
\beq
\label{H-check}
\check H_{j,n} = \gamma^n \check H_j \gamma^{-n}.
\eeq
Using that $\gamma T_{\La_j} \gamma^{-1} = q \La_j T_{\La_j}$, we obtain
$$
\check H_{1,n} = q^n\hr{\frac{\La_1^n}{1-\La_1/\La_2}T_{\La_1} + \frac{\La_2^n}{1-\La_2/\La_1}T_{\La_2}}
$$
and hence
$$
\check H_{1,n} \circ \Wc = \Wc \circ \hr{e^{\pi b c_b} \cdot \sigma_+^n e^{2\pi b x_1} \sigma_+^{-n}}.
$$
Note that by collecting coefficients in $t$, the first Macdonald operator $\Mtau_1$ can be expanded as a linear combination of twisted dual Toda operators:
$$
\Mtau_1 = t^{-1}\check H_{1,0} - q^{-2}\frac{t}{\La_1\La_2} \check H_{1,2}.
$$
As a consequence, we obtain the following intertwining relations:
\begin{align}
\label{eq:macdo-under-whit}
M_1^\tau \circ \Wc &= \Wc \circ \hr{e^{2\pi b(x_1-\tau)} + e^{2\pi b(x_2+\tau)} + e^{2\pi b(x_1+p_1-p_2+\tau)}}, \\
\nonumber M_2^\tau \circ \Wc &= \Wc \circ e^{2\pi b(x_1+x_2)}.
\end{align}
These relations can be interpreted as describing the action of the Macdonald operators $M_j^\tau$ in the basis of Whittaker functions $\Psi_{\bs\la}(\bs x)$, where the latter are regarded as functions of $\bs \lambda$.

\subsection{Spherical double affine Hecke algebra.}
Recall that the elliptic braid group is the fundamental group of the configuration space of $N$ points on the torus $T^2$. 
The double affine Hecke algebra $\Hbb_{q,t}$ for $GL(N)$  is the quotient of the $\mathbb{Q}(q,t)$-group algebra of the elliptic braid group by the Hecke relations
$$
(T_i - t)(T_i + t^{-1}) = 0 , \quad i=1,\ldots N,
$$
where $T_i$ are the standard ``local'' generators corresponding to those of the planar braid group. In addition to these local generators, there are also ``global'' ones, $X_i$ and $Y_i$, corresponding to the loops in the configuration space transporting the $i$-th marked point once along respectively the $(1,0)$- and the $(0,1)$-curves on the torus. Being a quotient of the elliptic braid group, the algebra $\Hbb_{q,t}$ carries an action of the modular group $SL(2;\mathbb{Z})$ by algebra automorphisms.

We skip the precise definition of the \emph{spherical subalgebra} $\SH_{q,t}$ of $\Hbb_{q,t}$, and refer the reader to \cite{Che05} for details. Instead, following \emph{loc.\!\! cit.\!} we recall that $\SH{q,t}$ admits a faithful representation into the algebra of symmetric $q$-difference operators in variables $(\Lambda_1,\Lambda_2)$. The spherical subalgebra contains elements $E_{(\pm k,0)}, E_{(0,\pm k)}$ corresponding to spherical versions of the $k$-th power sum symmetric functions in the standard generators $X^{\pm}_i,Y^{\pm}_i$ of $\Hbb_{q,t}$ respectively. As was shown in~\cite{SV11}, the algebra $\SH_{q,t}$ is generated over $\mathbb{Q}(q,t)$ by elements $E_{(\pm1,0)}, E_{(0,\pm1)}$, see \cite[Corollary 6.1]{BS12}.

Under Cherednik's representation, these generators are mapped to the following $q$-difference operators:
\begin{align*}
E_{(1,0)} &\longmapsto e_1(\bs\la), & E_{(-1,0)} &\longmapsto e_1(\bs\la)e_2(\bs\la)^{-1}, \\
E_{(0,1)} &\longmapsto \Mtau_1, & E_{(0,-1)} &\longmapsto \Mtau_1\hr{\Mtau_2}^{-1}.
\end{align*}
The action of $SL(2;\mathbb{Z})$ preserves the spherical subalgebra: for each $v\in\mathbb{Z}^2$, there is a corresponding element $E_v\in \SH_{q,t}$, and we have $g\cdot E_v = E_{gv}$ for all $g\in SL(2;\mathbb{Z}),~v\in\mathbb{Z}^2$.

\subsection{Strategy}
Let us now outline in more detail the logic of the remainder of the paper. The double affine Hecke algebra formalism suggests that the problem of diagonalizing the Macdonald difference operators can be solved by constructing a unitary representation $\rho$ of $SL(2;\mathbb{Z})$ on $L^2_{\mathrm{sym}}(\R^n, m(\bs\la)d\bs\la)$ compatible with its action on the spherical DAHA. Indeed, the action of the element $S$ would then intertwine the Macdonald operators with the operators of multiplication by elementary symmetric functions:
$$
\Mtau_j \circ \rho(S) = \rho(S) \circ e_j(\la).
$$
The eigenbasis for the latter is given by delta-distributions. Here we must note that while delta-distributions are clearly not in $L^2_{\mathrm{sym}}(\R^n, m(\bs\la)d\bs\la)$, it is only expected that the eigenfunctions of an unbounded operator lie in the space of tempered distributions rather than in the Hilbert space itself. Thus, the eigenbasis for the Macdonald operator would consist of the images of delta-distributions under the $S$ transformation:
$$
\Mtau_j \rho(S) \delta(\la-\mu) = e_j(\mu) \cdot \rho(S) \delta(\la-\mu).
$$

Our strategy for constructing the action of $SL(2;\mathbb{Z})$ on $L^2_{\mathrm{sym}}(\R^n, m(\bs\la)d\bs\la)$ is to identify the latter with the unitary equivalent representation $L^2(\R^n, d\bs x)$ via the Whittaker transform $\Wc$. As we will show, the action of generators of $\SH_{q,t}$ on $L^2(\R^n, d\bs x)$ coincides with that of certain elements of the quantized coordinate ring of a \emph{cluster Poisson variety}. The representation of the latter comes equipped with a compatible action of the group $SL(2;\mathbb{Z})$, which acts on $L^2(\R^n, d\bs x)$ by explicit unitary quantum cluster transformations. Thus the action $\rho(S)$ on $L^2_{\mathrm{sym}}(\R^n, m(\bs\la)d\bs\la)$ can be defined as
$$
\rho(S) = \Wc I_S \Wc^*,
$$
where $I_S$ denotes the action of $S$ on $L^2(\R^n, d\bs x)$, see~\eqref{diag}. Expressing the eigenfunctions of Macdonald operators as
$$
m(\la)^{-1} \rho(S) \delta(\la-\mu),
$$
where $m(\la)^{-1}$ is the normalization constant, and using the fact that
$$
\Wc^* \delta(\la-\mu) = \Psi_{\mu}(x),
$$
we see that the matrix coefficients $\ha{\Psi_\la, I_S \Psi_\mu}$ give an eigenbasis for the Macdonald operators.

\beq
\begin{tikzcd}
L^2_{\mathrm{sym}}(\R^n, m(\bs\la)d\bs\la) \arrow{rr}{\Wc^*} \arrow{dd}{\rho(S)} && L^2(\R^n, d\bs x) \arrow{dd}{I_S} \\ \\
L^2_{\mathrm{sym}}(\R^n, m(\bs\la)d\bs\la) && L^2(\R^n, d\bs x) \arrow[swap]{ll}{\Wc}
\end{tikzcd}
\label{diag}
\eeq

After reviewing the general construction of quantum cluster varieties in Section~\ref{sec:qcv}, the quantum cluster algebra needed to describe the spherical DAHA will be obtained from the moduli space of framed $GL(2)$-local systems on the punctured torus in Section~\ref{sec:cluster-torus}. In Section~\ref{sec:wave} we use the matrix coefficient presentation to rederive the formula for the Ruijsenaars wavefunctions along with some of its properties. Finally, in Section~\ref{sec:poly} we recover Macdonald polynomials from the Ruijsenaars wavefunctions.

%The double affine Hecke algebra formalism shows that the problem of diagonalizing the Macdonald difference operators boils down to constructing a unitary representation of $SL(2;\mathbb{Z})$ on $L^2_{\mathrm{sym}}(\R^n, m(\bs\la)d\bs\la)$ compatible with its action on the spherical DAHA. Indeed, the action of the element $S$ would then intertwine the Macdonald difference operators with the operators of multiplication by elementary symmetric functions. Our strategy for constructing the action of $SL(2;\mathbb{Z})$ on $L^2_{\mathrm{sym}}(\R^n, m(\bs\la)d\bs\la)$ is to use the Whittaker transform $\mathcal{W}$ to replace the latter with the unitary equivalent representation $L^2(\R^n, d\bs x)$. As we will show, this model for the representation has the advantage that the elements of $\SH_{q,t}$ act by elements of the quantized coordinate ring of a \emph{cluster Poisson variety}, and the group $SL(2;\mathbb{Z})$ acts on $L^2(\R^n, d\bs x)$ by explicit unitary quantum cluster transformations.  After reviewing the general construction of quantum cluster varieties in Section~\ref{sec:qcv}, the quantum cluster algebra needed to describe the spherical DAHA will be obtained from the moduli space of framed $GL(2)$-local systems on the punctured torus in Section~\ref{sec:cluster-torus}.

\section{Quantum cluster varieties}
\label{sec:qcv}
In this section we review the definition of quantum cluster varieties and their representations. For more details on these subjects, we refer the reader to the foundational paper~\cite{FG09}. We only need skew-symmetric quantum cluster algebras with integer-valued forms, which we incorporate in the definition of a seed.

\begin{defn}
A \emph{seed} is a datum $\Theta=\hr{I, I_0, \La, (\cdot,\cdot),\hc{e_i}}$ where
\begin{itemize}
\item $I$ is a finite set;
\item $I_0 \subset I$ is a \emph{frozen} subset of $I$;
\item $\La$ is a lattice;
\item $(\cdot,\cdot)$ is a skew-symmetric $\Z$-valued form on $\La$;
\item $\hc{e_i \,|\, i \in I}$ is a basis for the lattice $\La$.
\end{itemize}
Note that the data of the last point is equivalent to that of an isomorphism $\bs e \colon \Z^{I} \simeq \La$. In particular, given a pair of seeds $(\Theta,\Theta')$ with the same data $\hr{I,I_0}$, we get a canonical isomorphism of abelian groups (not necessarily isometry of lattices) $\bs{e'} \bs e^{-1} \colon \La \simeq \La'$.
\end{defn}

\begin{defn}
 We say that $(\Theta,\Theta')$ are \emph{equivalent} if the isomorphism $\bs{e'} \bs e^{-1} \colon \La \simeq \La'$ is an isometry, that is $(e_i,e_j)_{\La} = (e'_i, e'_j)_{\La'}$ for all $i,j \in I$. We define a \emph{quiver} to be an equivalence class of seeds.
\end{defn}

The quiver $Q$ associated to a seed $\Theta$ can be visualized as a directed graph with vertices labelled by the set $I$ and arrows given by the adjacency matrix $\eps = \hr{\eps_{ij}}$, where $\eps_{ij} = (e_i,e_j)$. The vertices corresponding to the subset $I_0$ are called \emph{frozen} and are drawn as squares, while those corresponding to $I \setminus I_0$ are referred to as \emph{mutable} and are drawn as circles.
 
The pair $\hr{\Lambda,(\cdot, \cdot)}$ determines a \emph{quantum torus algebra} $\mathcal{T}_\Lambda^q$, which is defined to be the free $\Z[q^{\pm1}]$-module spanned by $\hc{Y_\la \,|\, \la\in\La}$, with the multiplication defined by
$$
q^{(\lambda,\mu)}Y_\lambda Y_\mu = Y_{\lambda+\mu}.
$$
A basis $\hc{e_i}$ of the lattice $\La$ gives rise to a distinguished system of generators for $\mathcal{T}_\Lambda^q$, namely
%$$
%\Tc^q_\La \simeq T^q_Q = \Z[q^{\pm1}]\ha{X_i \,|\, i \in I} / \ha{X_j X_i = q^{2\eps_{ij}} X_i X_j}
%$$
the elements $Y_i=Y_{e_i}$. This way we obtain a \emph{quantum cluster $\Xc$-chart}
\begin{align}
\label{eq:qtor-presentation}
\Tc_Q^q = \Z[q^{\pm1}]\ha{Y_i^{\pm1} \,|\, i \in I} / \ha{q^{\eps_{jk}}Y_jY_k = q^{\eps_{kj}}Y_kY_j} \simeq \Tc_\La^q.
\end{align}
The generators $Y_i$ are the \emph{quantum cluster $\Xc$-variables}. We note that this presentation of $\Tc_Q^q$ depends only on the quiver and not on the choice of the representative seed. 

A seed also determines a Heisenberg algebra
$$
\heis_\La=\La \otimes_\Z \R \oplus \R c
$$
with generators $\hc{y_\la \,|\, \la\in\La}$ and $c$, Lie bracket
$$
[y_\la,y_\mu]=\hr{\la,\mu}c,
$$
and the $*$-structure
$$
*y_\la=y_\la, \qquad *c=-c.
$$
If the form $\hr{\cdot,\cdot}$ is non-degenerate, the Lie algebra $\heis_\La$ has an irreducible $*$-representation $\rho$ on a Hilbert space $\mathcal{H}$ in which the generators $y_\lambda$ act by unbounded, essentially self-adjoint operators, and the central element $c$ acts by the scalar $1/2\pi i$ with $i = \sqrt{-1}$. We write $\mathrm{Heis}_{\Lambda}$ for the quotient of the universal enveloping algebra $U(\heis_\La)$ by the ideal $\ha{2\pi i c-1}$.

The assignment
$$
Y_\la = e^{2\pi\hbar y_\la}
$$
defines an embedding of $\mathcal{T}^q_Q$ into the set of grouplike elements in $\mathrm{Heis}^\hbar_{\Lambda}$, the central quotient of the $\hbar$-adically completed universal enveloping algebra $U(\heis_\La)$. In the Hilbert space representation $\mathcal{H}$, the quantum torus generators $Y_{\lambda}$ act by unbounded, positive, essentially self-adjoint operators.

If $\Theta,\Theta'$ are two seeds with nondegenerate skew forms representing the same quiver, then the isometry $\Lambda\rightarrow \Lambda'$, $e_i\mapsto e_i'$ determines a canonical isomorphism of Heisenberg algebras $\iota \colon \heis_\La \to \heis_{\La'}$. So by the irreducibility of the canonical representations $\mathcal{H},\mathcal{H}'$, there is a \emph{unique} projective unitary equivalence $\mathbb I \colon \Hc \to \Hc'$ such that 
\beq
\label{eq:unique-iso}
\rho'(\iota(a))\circ \mathbb I= \mathbb I \circ \rho(a) \qquad\text{for all}\qquad a\in\heis_\La,
\eeq
where both sides are understood as maps between the corresponding classical Schwartz spaces in $\mathcal{H}_\Theta,\mathcal{H}_{\Theta'}$. This shows that the data $(\La,\heis_\La,\Hc)$ associated to a nondegenerate quiver is unique up to unique isomorphism. Because of this, we will often abuse language and speak of `the' Heisenberg algebra, or Hilbert space, associated to a given quiver.

When the form $\hr{\cdot,\cdot}$ has nonzero kernel $Z$, the Heisenberg algebra has a family of irreducible $*$-representations $\mathcal{H}_{\chi}$ parametrized by central characters $\chi \colon Z \to \R$. In this case, we need to enrich the data of a quiver in order to unambiguously speak of its Heisenberg algebra and representation as explained above. Indeed, unless the central characters are compatible, the isomorphism in~\eqref{eq:unique-iso} cannot exist. 
\begin{defn}
A \emph{seed with central character} is a pair $(\Theta,\chi)$ where $\Theta$ is a seed and $\chi \colon Z\to\R$ is a linear functional on the kernel $Z$ of the corresponding skew-form. We say that two seeds $(\Theta,\chi)$ and $(\Theta',\chi')$ with central character are equivalent if the underlying seeds are equivalent, and the canonical isometry $\iota \colon \Lambda\rightarrow\Lambda'$ of lattices intertwines the central characters: $\hr{\chi'\circ\iota}|_Z = \chi$. A \emph{quiver with central character} is an equivalence class of seeds with central character.
\end{defn}

Let $\Theta,\Theta'$ be seeds representing quivers $Q,Q'$ (possibly with central characters). We say that the quiver $Q'$ is the \emph{mutation of $Q$ in direction $k\in I\setminus I_0$} if the map 
\beq
\label{eq:mon-mut}
\mu_k \colon \Lambda \longra \Lambda', \qquad e_i \longmapsto 
\begin{cases}
-e'_k &\text{if} \; i=k, \\
e'_i + \max\{(e_i,e_k),0\}e'_k &\text{if} \; i \ne k
\end{cases}
\eeq
is an isometry intertwining the central characters. It is easy to see that $Q'=\mu_k(Q)$ if and only if $ Q=\mu_k(Q')$.
%
%Let $\Theta$ be a seed, and $k \in I \smallsetminus I_0$ a mutable vertex of the corresponding quiver $Q$. Then one obtains a new seed, $\mu_k(\Theta)$, called the \emph{mutation of $\Theta$ in direction $k$}, by changing the basis~$\hc{e_i}$ while the rest of the data remains the same. The new basis $\{e_i'\}$ is defined by
%$$
%e'_i = 
%\begin{cases}
%-e_k &\text{if} \; i=k, \\
%e_i + [\eps_{ik}]_+e_k &\text{if} \; i \ne k,
%\end{cases}
%$$
%where $[a]_+=\max(a,0)$. Although the bases $\hc{e_i}$ and $\hc{\mu_k^2(e_i)}$ do not necessarily coincide, the seeds $\Theta$ and $\mu_k^2(\Theta)$ are equivalent. Hence the above formula descends to give a  well-defined notion of the mutation of the quiver in direction $k$. 
The \emph{mutation class} of a quiver $Q$, which we denote by the bold symbol $\bs Q$, is the set of all quivers related to $Q$ by some finite sequence of mutations. 

To each quiver mutation $\mu_k$ we associate an isomorphism of quantum tori
$$
\mu'_k \colon \Tc_Q^q \longra \Tc_{\mu_k(Q)}^q,
$$
and define the \emph{quantum cluster $\Xc$-mutation}
$$
\mu^q_k \colon \Frac(\Tc_{Q}^q) \longra  \Frac(\Tc_{Q'}^q), \qquad f \longmapsto \Psi_q\hr{Y_k'} \mu'_k(f) \Psi_q\hr{Y_k'}^{-1}
%\mu^q_k =  \Ad_{\Psi_q\hr{Y_k'}}\circ\mu_k'.
$$
where $\Frac(\Tc_Q)$ denotes the skew fraction field of the Ore domain $\Tc_Q$, and $\Psi_q(Y)$ is the (compact) quantum dilogarithm function~\eqref{eq:compact-qdl}.
%
%\red{$$
%\mu^q_k = \Ad_{\Psi_q\hr{Y_{-e_k}}},
%$$}
The fact that conjugation by $\Psi_q\hr{Y'_{k}}$ yields a birational automorphism is guaranteed by the integrality of the form~$(\cdot, \cdot)$ and functional equations~\eqref{eqn:q-Gamma}.

\begin{defn}
An element of $\Tc^q_Q$ is said to be \emph{universally Laurent} if its image under any finite sequence of quantum cluster mutations 
%(which is {\it a priori} only a fraction field element) 
is contained in the corresponding quantum torus algebra.  
%Laurent polynomial in cluster $\Xc$-variables under any sequence of cluster mutations. 
The \emph{universally Laurent algebra} $\Lbb^q_{\bs Q}$ is the algebra of universally Laurent elements of $\Tc^q_Q$.
\end{defn}
{If $q$ is a complex number, we can also consider the $\mathbb{C}$-algebra given by the corresponding specialization of $\Lbb^q_{\bs Q}$. We abuse notation and denote this algebra by the same symbol.}

\begin{defn}
Let $b\in\mathbb{R}_{>0}$ and consider the specializations
$$
\hbar = b, \quad q = e^{\pi ib^2},\quad \tilde q = e^{\pi ib^{-2}},
$$
so that $q,\tilde q$ are complex numbers on the unit circle.
The algebra
$$
\Lbb_{\bs Q} = \Lbb^{q, \tilde q}_{\bs Q} = \Lbb^q_{\bs Q} \otimes_{\mathbb{C}} \Lbb^{\tilde q}_{\bs Q}
$$
is called the \emph{modular double} of the universally Laurent algebra $\Lbb^q_{\bs Q}$. The maximal joint domain $\mathcal{S}_Q\subset\mathcal{H}_Q$ for the action of $\Lbb_{\bs Q}$ is called the \emph{Fock-Goncharov Schwartz space}.
\end{defn}

\begin{remark}
In a more general setup, when the bilinear form $(\cdot,\cdot)$ is not required to be integer-valued, the cross-relations between elements of $\Tc^q_Q$ and $\Tc^{\tilde q}_Q$ are set to be
$$
Y_j \widetilde Y_k = e^{2 \pi i \eps_{kj}} \widetilde Y_k Y_j.
$$
\end{remark}

The collection of quantum charts $\Tc^q_Q$ with $Q \in \bs Q$, together with quantum cluster $\Xc$-mutations is often referred to as the \emph{quantum cluster $\Xc$-variety}. We regard the quantum charts as the quantized algebras of functions on the toric charts in the atlas for the classical cluster Poisson variety.
%, isomorphic to $(\C^\times)^{|I|}$, endowed with log-canonical Poisson brackets
%$$
%\hc{Y_j, Y_k} = \eps_{jk} Y_j Y_k.
%$$
The quantum charts form an $\ell$-regular tree with $\ell = |I \setminus I_0|$, and the cluster mutations quantize the gluing data between adjacent charts. The universally Laurent algebra is the quantum analog of the algebra of global functions on the cluster variety. Unless otherwise specified in what follows, we will simply write `cluster variety' for quantum cluster $\Xc$-variety --- the same applies to variables, charts, mutations, etc.

The \emph{(quasi-)cluster modular groupoid} associated to a cluster variety is defined as follows. 
%A pair of cluster charts are said to be \emph{equivalent} if the corresponding quivers are isomorphic. The latter happens if and only if there exists a permutation of variables which establishes an isomorphism of the corresponding quantum cluster charts.

%\begin{defn}
%An automorphism of $\Tc_\La^q$ is called a \emph{quasi-equivalence} if it can be presented as a composition of a permutation of cluster variables and a map, which is identity on mutable variables and sends each frozen variable to a monomial.
%\end{defn}
\begin{defn}
Let $Q,Q'$ be two quivers with the same label sets $(I,I_0)$.  We define a \emph{generalized permutation} to be a monomial isomorphism of quantum tori $\vartheta \colon \Tc_{Q}\rightarrow \Tc_{Q'}$ such that 
for some permutation $\sigma \colon (I\setminus I_0)\rightarrow (I\setminus I_0)$ of the non-frozen directions we have $\vartheta(Y_i)  = Y'_{\sigma(i)}$ for all $i\in I\setminus I_0$. If the quivers $Q,Q'$ are equipped with central characters, we require these to be intertwined by the monomial map $\vartheta$.
\end{defn}
In particular, the non-frozen submatrices of the matrices $\eps$, $\eps'$ for quivers $Q$, $Q'$ related by a generalized permutation $\vartheta$ are conjugate under an actual permutation $\sigma$ of the non-frozen set $I\setminus  I_0$. 

By the same construction as~\eqref{eq:unique-iso} with $\vartheta$ playing the role of $\iota$, we associate to each quasicluster transformation $\vartheta$ a projective unitary transformation
$
\mathbb{I}_\vartheta \colon \mathcal{H}_Q\longra\mathcal{H}_{Q'}.
$

%\red{Gus comment: I got confused by the name `quasi-equivalence' because it plays a role much more like the permutations (at the level of quivers) than the previous equivalences (at the level of seeds).}
%%
%We say that the two cluster charts are quasi-equivalent of their images in $\Tc_\La^q$ are. Furthermore, we deem a pair of quivers quasi-equivalent when the corresponding charts are.

%\begin{defn}
%A \emph{quasi-cluster transformation} is an automorphism of $\Frac(\Tc_\La^q)$ which may be presented as a composition of cluster mutations and a quasi-equivalence.
%\end{defn}

\begin{defn}
Let $Q,Q'$ be two quivers with corresponding quantum tori $\Tc_Q,\Tc_Q'$ as in~\eqref{eq:qtor-presentation}. A \emph{quasi-cluster transformation} with source $Q$ and target $Q'$ is a non-commutative birational isomorphism $\Tc_Q\dashrightarrow \Tc_Q'$  which can be factored as a composition of cluster mutations and generalized permutations.
\end{defn}

\begin{defn}
The \emph{quasi-cluster modular groupoid} is the groupoid $\Gc_{\bs Q}$ whose objects are quivers $Q \in \bs Q$ {(possibly with central characters)}, and whose morphisms are quasi-cluster transformations. The \emph{quasi-cluster modular group}, denoted $\Gamma_{\bs Q}$, is the automorphism group of an object in $\Gc_{\bs Q}$.
\end{defn}
\begin{remark}
Any element of the quasi-cluster modular group restricts to an automorphism of the universally Laurent algebra $\Lbb_{\bs Q}$.
\end{remark}

Denote by $\Hilb$ the category whose objects are Hilbert spaces and whose morphisms are projective unitary transformations. At this point, we have associated to each object $Q$ in the groupoid $\Gc_{\bs Q}$ a Hilbert space $\mathcal{H}_Q$, and associated unitary intertwiners $\mathbb{I}_\vartheta$ to those morphisms in $\Gc_{\bs Q}$ given by generalized permutations $\vartheta$. In~\cite{FG09}, it was further shown how to associate intertwiners $\mathbb{I}_{\mu_k}$ to the mutation morphisms in $\Gc_{\bs Q}$, which send $\mathcal{S}_Q$ to $\mathcal{S}_{Q'}$ and give rise to well-defined functor $\mathbb{I} \colon \Gc_{\bs Q}\rightarrow \Hilb$, that is a \emph{projective representation} of the groupoid $\Gc_{\bs Q}$.

Let us briefly recall the construction of the intertwiner corresponding to a mutation $\mu_k \colon Q\rightarrow Q'$. Once again, the construction ~\eqref{eq:unique-iso} with the monomial isomorphism $\mu_k'$ playing the role of $\iota$ provides a projective unitary transformation 
$
\mathbb{I}_{\mu_k'} \colon \mathcal{H}_Q\rightarrow\mathcal{H}_{Q'}. 
$
On the other hand, by the property~\eqref{eq:qdl-unitary} of the quantum dilogarithm function $\varphi(z)$ it follows that any self-adjoint operator $y$ on $\mathcal{H}$ defines a unitary operator $\varphi(y) \colon \mathcal{H}\simeq\mathcal{H}$. The unitary transformation $\mathbb{I}_{\mu_k} \colon \mathcal{H}_Q \to \mathcal{H}_{Q'}$ is then defined to be
$$
\mathbb{I}_{\mu_k} = \mathbb{I}_{\mu_k'}\circ \varphi(-y_k)^{-1}=\varphi(y'_k)^{-1}\circ \mathbb{I}_{\mu_k'}.
$$
That this recipe is indeed well-defined is a difficult theorem, which is the main analytical result of the paper~\cite{FG09}. In the following section, we will illustrate how to use it in practice with a concrete example.

\section{Quantum Teichm\"uller theory for the punctured torus}
\label{sec:cluster-torus}
In this section we illustrate the previous definitions in our main example, the moduli space of framed $GL_2$ local systems on a punctured torus.

\subsection{Quantized moduli space of framed $GL_2$-local systems on a punctured torus.}

Consider the quiver $Q$ shown on Figure~\ref{fig:quiver} and the associated cluster variety. The kernel of the corresponding skew-form is spanned by the vector $z = -e_1-e_2-e_3$. For $\tau\in\mathbb{R}$, we write $Q_\tau$ for the quiver with central character defined by $\mathcal{X}_\tau(z)={2}\tau$. The latter corresponds to a relation
$$
Y_1Y_2Y_3 = q^2e^{-4\pi b \tau}.
$$

\begin{figure}[h]
\begin{tikzpicture}[every node/.style={inner sep=0, minimum size=0.45cm, thick, draw, circle}, thick, x=1cm, y=0.866cm]

\node (1) at (0,0) {\footnotesize{1}};
\node (2) at (-1,2) {\footnotesize{2}};
\node (3) at (1,2) {\footnotesize{3}};
\node[rectangle] (4) at (-2,0) {\footnotesize{4}};
\node[rectangle] (5) at (2,0) {\footnotesize{5}};

\draw [->] (3.165) -- (2.15);
\draw [->] (1.45) -- (3.-105);
\draw [->] (1.75) -- (3.-135);
\draw [->] (2.-45) -- (1.105);
\draw [->] (2.-75) -- (1.135);
\draw [->] (3.-165) to (2.-15);

\draw[->] (1) to (4);
\draw[->] (4) to (2);
\draw[->] (3) to (5);
\draw[->] (5) to (1);

\end{tikzpicture}
\caption{Quiver $Q$.}
\label{fig:quiver}
\end{figure}
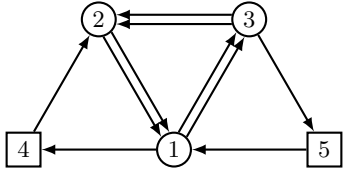

The mutable part of the quiver is known as the \emph{Markov quiver}, and is isomorphic to any of its mutations. The Markov quiver describes the cluster structure on the moduli space of framed $SL_2$-local systems on a punctured torus, where the framing data consists of a flag (a line in rank 1), invariant under the monodromy around the puncture, see~\cite{FG06}. The additional frozen vertices will be used to parametrize the determinants of holonomies around loops on the punctured torus, as we now explain.

At the quantum level, the algebra $\Lbb_{\bs Q}$ contains elements corresponding to the traces and determinants of holonomies around closed simple curves. For $(m,n) \in \Z^2$, we denote by $L_{(m,n)}$ and $\Delta_{(m,n)}$ the elements corresponding respectively to the trace and the determinant of the monodromy along the $(m,n)$-curve on the punctured torus. These elements can be constructed explicitly as follows.
%Recall that the mutable part of the quiver $Q$ describes the surface cluster algebra of a punctured torus, see Figure~\ref{fig:torus}. 
We choose a basis in $H_1(T^2\setminus D^2;\Z)\simeq \mathbb{Z}^2$ so that the horizontal cycle, passing from left to right through nodes 1 and 2 in Figure~\ref{fig:torus}, is of homology class $(1,0)$. Then we set
% {\color{red} (refer to FG, Teschner \& co, or us)}
\beq
\label{eq:SD10}
L_{(1,0)} = Y_{e_4} + Y_{e_2+e_4} + Y_{e_1+e_2+e_4}, \qquad \Delta_{(1,0)} = Y_{e_1+e_2+2e_4}.
\eeq
As was shown in~\cite[Proposition 3.4]{SS19}, elements $L_{(1,0)}$ and $\Delta_{(1,0)}$ are universally Laurent:
$$
L_{(1,0)}, ~\Delta_{(1,0)} \in \Lbb_{\bs Q}.
$$

\begin{figure}[h]
\begin{tikzpicture}[every node/.style={inner sep=0, minimum size=0.45cm, thick, draw, circle}, thick, x=0.8cm, y=0.8cm]

\draw[very thick, gray!50] (-2,-2) to (-2,2) to (2,2) to (2,-2) to (-2,-2) to (2,2);

\node[fill=white] (1) at (0,0) {\footnotesize{1}};
\node[fill=white] (2_1) at (-2,0) {\footnotesize{2}};
\node[fill=white] (2_2) at (2,0) {\footnotesize{2}};
\node[fill=white] (3_1) at (0,2) {\footnotesize{3}};
\node[fill=white] (3_2) at (0,-2) {\footnotesize{3}};

\draw [->] (1) -- (3_1);
\draw [->] (3_1) -- (2_1);
\draw [->] (2_1) -- (1);
\draw [->] (1) -- (3_2);
\draw [->] (3_2) -- (2_2);
\draw [->] (2_2) -- (1);

\end{tikzpicture}
\caption{Cluster quiver from a punctured torus.}
\label{fig:torus}
\end{figure}
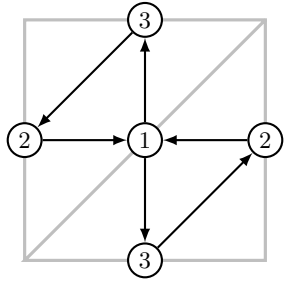

Formulas~\eqref{eq:SD10} can be written down with the help of the following standard combinatorial recipe, see e.g.~\cite[Section 9]{FG06}. Consider the bipartite graph on a cylinder, shown on the left of Figure~\ref{fig:network}. On the right we show a dual quiver, which we use to define a cluster variety, with edges directed in such a way that the white vertex of the bipartite graph stays on the right as we traverse an edge. The direction of edges of the bi-partite graph is additional data, which allows one to express a monodromy matrix $M$ in cluster coordinates. Namely, we set
$$
M_{ij} = \sum_{p : \, i \to j} Y_{\wt(p)}, \qquad\text{where}\qquad \wt(p) = \sum_{f \, \text{below} \, p} e_f.
$$
The first sum in the formula above is taken over all paths in the directed bipartite graph from $i$-th source to the $j$-th sink, while the second is taken over all faces lying below the path $\wt(p)$. Note that the matrix $M$ only depends on variables 1, 2, and 4, but not on the variable 0. Since the monodromy matrix is only defined up to conjugation, we shall be looking at the symmetric functions of its eigenvalues: the trace and the determinant in our case, which are equal to $L_{1,0}$ and $\Delta_{1,0}$ respectively. Let us also note that in order to work with $SL_2$ rather than $GL_2$ local systems one only needs to specialize $y_4 = -\frac12(y_1+y_2)$.

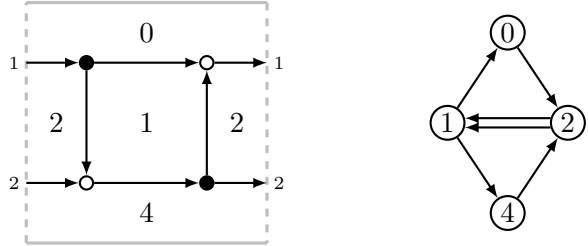
\begin{figure}[h]
\begin{tikzpicture}[every node/.style={inner sep=0, minimum size=0.45cm, thick, draw, circle}, thick, x=0.8cm, y=0.8cm]

\draw[very thick, dashed, gray!50] (-2,-2) to (-2,2);
\draw[very thick, gray!50] (-2,2) to (2,2);
\draw[very thick, dashed, gray!50] (2,2) to (2,-2);
\draw[very thick, gray!50] (2,-2) to (-2,-2);

\node[fill=white, draw=none] at (0,1.5) {0};
\node[fill=white, draw=none] at (0,0) {1};
\node[fill=white, draw=none] at (-1.5,0) {2};
\node[fill=white, draw=none] at (1.5,0) {2};
\node[fill=white, draw=none] at (0,-1.5) {4};

\node[fill=white, minimum size=5pt] (p1p1) at (1,1) {};
\node[fill=white, minimum size=5pt] (m1m1) at (-1,-1) {};
\node[fill, minimum size=5pt] (m1p1) at (-1,1) {};
\node[fill, minimum size=5pt] (p1m1) at (1,-1) {};

\draw[->] (-2,1) to (m1p1);
\draw[->] (m1p1) to (p1p1);
\draw[->] (p1p1) to (2,1);
\draw[->] (-2,-1) to (m1m1);
\draw[->] (m1m1) to (p1m1);
\draw[->] (p1m1) to (2,-1);
\draw[->] (m1p1) to (m1m1);
\draw[->] (p1m1) to (p1p1);

\node[fill=white] (0) at (6,1.5) {0};
\node[fill=white] (1) at (5,0) {1};
\node[fill=white] (2) at (7,0) {2};
\node[fill=white] (4) at (6,-1.5) {4};

\draw[->] (2.165) to (1.15);
\draw[->] (2.-165) to (1.-15);
\draw[->] (1) to (4);
\draw[->] (4) to (2);
\draw[->] (1) to (0);
\draw[->] (0) to (2);

\node[fill=white, draw=none, minimum size=0] at (-2.2,1) {\tiny 1};
\node[fill=white, draw=none, minimum size=0] at (-2.2,-1) {\tiny 2};
\node[fill=white, draw=none, minimum size=0] at (2.2,1) {\tiny 1};
\node[fill=white, draw=none, minimum size=0] at (2.2,-1) {\tiny 2};

\end{tikzpicture}
\caption{Directed network and dual cluster quiver.}
\label{fig:network}
\end{figure}

%Now let $G$ be a simple Lie group, $\Sigma$ a marked surface, as in~\cite{FG09}, and $\Xc_{G,\Sigma}$ a quantum cluster variety arising from the moduli space of $G$-local systems on $\Sigma$. Then the mapping class group $\Gamma_\Sigma$ of the surface $\Sigma$ admits an injective  homomorphism into the modular cluster group of $\Xc_{G,\Sigma}$. Let us illustrate it on the example at hand. 
The mapping class group of a punctured torus is isomorphic to $SL(2,\Z)$, and is generated by elements
$$
\sigma_+ =
\begin{pmatrix}
1 & 1 \\
0 & 1
\end{pmatrix}
\qquad\text{and}\qquad
\sigma_- =
\begin{pmatrix}
1 & 0 \\
1 & 1
\end{pmatrix}
$$
which correspond to the Dehn twists of the torus along closed simple curves with homology classes $(1,0)$ and $(0,1)$ respectively. By the construction in~\cite[Section 6]{FG09}, we get a homomorphism $SL(2,\Z) \to \Gamma_{\bs Q_\tau}$ sending
$$
\sigma_+^{-1} \longmapsto m_+ \circ \mu_1^q, \qquad
\sigma_- \longmapsto m_- \circ \mu_3^q,
$$
where $m_+ \colon \mathcal{T}_{\mu_1(Q_\tau)} \to \mathcal{T}_{Q_\tau}$ and $m_- \colon \mathcal{T}_{\mu_3(Q_\tau)} \to \mathcal{T}_{Q_\tau}$
 are the following generalized permutations:
 \beq
\label{eq:qcmplus}
\begin{aligned}
m_+ &= \hc{e'_1 \mapsto e_2, \, e'_2 \mapsto e_1, \, e'_3 \mapsto e_3, \, e'_4 \mapsto e_4, \, e'_5 \mapsto e_1+e_4+e_5}, & &e'_j = \mu_1(e_j), \\
m_- &= \hc{e'_1 \mapsto e_3, \, e'_2 \mapsto e_2, \, e'_3 \mapsto e_1, \, e'_4 \mapsto -e_1-e_3+e_4-e_5, \, e'_5 \mapsto e_5}, & &e'_j = \mu_3(e_j).
\end{aligned}
\eeq
The element $\sigma$ of order 6 defined by
$$
\sigma = \sigma_+^{-1}\sigma_- = \begin{pmatrix} 0 & -1 \\ 1 & 1 \end{pmatrix}
$$
is mapped under this homomorphism to a generalized permutation:
$$
\sigma \longmapsto \hc{e_1 \mapsto e_3, \, e_2 \mapsto e_1, \, e_3 \mapsto e_2, \, e_4 \mapsto -e_1-e_3-e_5, \, e_5 \mapsto e_1+e_4+e_5}.
$$
Most important for us will be the order 4 element
$$
S = \sigma_+^{-1}\sigma_-\sigma_+^{-1} = \begin{pmatrix} 0 & -1 \\ 1 & 0 \end{pmatrix}.
$$
Given $g = \begin{pmatrix} a & b \\ c & d \end{pmatrix} \in SL(2,\Z)$ we define 
\beq
\label{Sac}
L_{(a,c)} = g\cdot L_{(1,0)}, \qquad \Delta_{(a,c)} = g\cdot \Delta_{(1,0)}
\eeq
The above definitions make sense thanks to the fact that
$$
\sigma_+\hr{L_{(1,0)}} = L_{(1,0)}
\qquad\text{and}\qquad
\sigma_+\hr{\Delta_{(1,0)}} = \Delta_{(1,0)}.
$$
For example, the element $S$ satisfies
\beq
\label{10to01}
S\hr{L_{(1,0)}} = L_{(0,1)}, \qquad S\hr{\Delta_{(1,0)}} = \Delta_{(0,1)}.
\eeq

Note that $g L_{(1,0)}$ is universally Laurent for any $g \in SL(2,\Z)$ since $L_{(1,0)}$ is. Using~\eqref{Sac} we now compute
\begin{align*}
L_{(0,1)} &= Y_{e_5}^{-1} + Y_{e_3+e_5}^{-1} + Y_{e_1+e_3+e_5}^{-1}, & \Delta_{(0,1)} &= Y_{e_1+e_3+2e_5}^{-1},
\end{align*}
as well as
\begin{align*}
L_{(0,-1)} &= Y_{e_5} + Y_{e_1+e_5} + Y_{e_1+e_3+e_5}, & \Delta_{(0,-1)} &= Y_{e_1+e_3+2e_5}, \\
L_{(1,-1)} &= Y_{e_1+e_4+e_5} + Y_{e_1+e_3+e_4+e_5} + Y_{e_1+e_2+e_3+e_4+e_5}, & \Delta_{(1,-1)} &= Y_{2e_1+e_2+e_3+2e_4+2e_5}.
\end{align*}
Finally we observe that for any $a,b,c,d \in \Z$, satisfying $ad-bc=1$, we have
$$
\Delta_{(b,d)}\Delta_{(-b,-d)} = 1, \qquad
L_{(-b,-d)} = \Delta_{(b,d)}^{-1}L_{(b,d)}, \qquad
\hs{L_{(a,c)},L_{(b,d)}} = {\hr{q^{-1}-q}} L_{(a+c,b+d)}.
$$
Indeed, the above relations hold for $(a,b,c,d) = (0,1,-1,0)$, and the general case follows from~\eqref{Sac}.

\subsection{Unitary representation of $SL(2,\mathbb{Z})$}
Following the construction of Section~\ref{sec:qcv}, we consider the Hilbert space $\mathcal{H}_Q$ associated to the quiver $Q_\tau$ from Figure~\eqref{fig:quiver}. In the approach to quantum Teichm\"uller theory of~\cite{FG09}, the  Hilbert space $\mathcal{H}_Q$ is the one assigned to the punctured torus.

We now compute the action of the generators $\sigma_\pm$ on $\mathcal{H}_Q$.  In order to write explicit formulas, let us fix the isomorphism $\mathcal{H}_{Q_\tau}\simeq L_2(\mathbb{R}^2)$ such that the action of the Heisenberg generators on a Schwartz function $f(x_1,x_2)$ is given by
%
%
%Choosing the following polarization
\beq
\label{eq:polarization}
\begin{aligned}
&y_1 \longmapsto p_1-p_2+x_1-x_2, & &\qquad y_2 \longmapsto x_2-x_1, & &\qquad y_3 \longmapsto p_2-p_1-2\tau, \\
&y_4 \longmapsto p_2, & &\qquad y_5 \longmapsto -x_1+\tau,
\end{aligned}
\eeq
where
$$
p_j = \frac{1}{2\pi i}\frac{\partial}{\partial x_j}.%, \quad [p_j,x_k]=\frac{\delta_{j,k}}{2\pi i}.
$$
Under this isomorphism, the quantum holonomies act by the operators
\beq
\label{eq:L-action}
\begin{aligned}
L_{(1,0)} &\longmapsto e^{2\pi b p_2} + e^{2\pi b(p_2+x_2-x_1)} + e^{2\pi bp_1}, & \Delta_{(1,0)} &\longmapsto e^{2\pi b(p_1+p_2)}, \\
L_{(0,1)} &\longmapsto e^{2\pi b(x_1-\tau)} + e^{2\pi b(p_1-p_2+x_1+\tau)} + e^{2\pi b(x_2+\tau)}, & \Delta_{(0,1)} &\longmapsto e^{2\pi b(x_1+x_2)}.
\end{aligned}
\eeq
It is then easy to check that the projective transformations $\mathbb{I}_{\sigma_\pm} \colon \mathcal{H}_{Q_\tau}\rightarrow\mathcal{H}_{Q_\tau}$ take the following form:
\beq
\label{eq:dehn-actions}
\begin{aligned}
\mathbb{I}_{\sigma_+^{-1}} &= e^{-\pi i(p_1^2+p_2^2)} \varphi(p_2-p_1+x_2-x_1)^{-1},\\
\mathbb{I}_{\sigma_-} &= e^{2\pi i \tau(x_1-x_2)} e^{-\pi i(x_1^2+x_2^2)} \varphi(p_1-p_2+2\tau)^{-1}.
\end{aligned}
\eeq
To illustrate the construction, we spell out the derivation of the projective transformation $\sigma_+^{-1}$. Composing the maps~\eqref{eq:mon-mut} and~\eqref{eq:qcmplus}, we find that the monomial automorphism $m_+\circ \mu_1'$ of $\mathfrak{heis}_{Q_\tau}$ is given by
$$
\begin{aligned}
&y_1 \longmapsto -y_2, & &\qquad y_2 \longmapsto y_1+2y_2, & &\qquad y_3 \longmapsto y_3, \\
&y_4 \longmapsto y_4, & &\qquad y_5 \longmapsto y_1+y_2+y_4+y_5.
\end{aligned}
$$
The monomial part $\mathbb{I}_{m_+}\circ \mathbb{I}_{\mu_1'}$ of the intertwiner $\mathbb{I}_{\sigma_+^{-1}}$ is then defined as the unique projective transformation satisfying the intertwining relations
$$
(\mathbb{I}_{m_+}\circ \mathbb{I}_{\mu_1'})\circ y_k = (m_+\circ \mu_1')(y_k)\circ (\mathbb{I}_{m_+}\circ \mathbb{I}_{\mu_1'}) \qquad\text{for all}\qquad 1\le k\le 5.
$$
It is straightforward to check that under the isomorphism~\eqref{eq:polarization}, these relations are satisfied by the operator $e^{-\pi i(p_1^2+p_2^2)}$. The formula for $\sigma_-$ is obtained in the same way.

%Using the inversion formula~\eqref{inv} to rewrite the latter in the equivalent form
Rewriting the latter expression as
$$
\mathbb{I}_{\sigma_-} = \varphi(p_1-p_2+x_1-x_2)^{-1}  e^{2\pi i \tau(x_1-x_2)} e^{-\pi i(x_1^2+x_2^2)},
$$
we may use the inversion formula~\eqref{inv} to see that the element $\sigma$ is mapped to the projective transformation
%we see that the element $\sigma$ is mapped to the projective transformation
\beq
\label{spm}
\mathbb{I}_{\sigma} \longmapsto e^{2\pi ix_1x_2} \Fc e^{2\pi i\tau(x_1-x_2)}.
\eeq
Here $\Fc$
%\begin{align*}
%\Fc &=\red{i}  e^{-\pi i(x_1^2+x_2^2)} e^{-\pi i(p_1^2+p_2^2)} e^{-\pi i(x_1^2+x_2^2)}\\
%&= \red{i}e^{-\pi i(p_1^2+p_2^2)} e^{-\pi i(x_1^2+x_2^2)} e^{-\pi i(p_1^2+p_2^2)}
%\end{align*}
is the Fourier transform that may be written as
$$
\Fc \colon f(\bs x) \longmapsto \int_{\R^2} f(\bs y) e^{-2\pi i \bs x \cdot \bs y} d\bs y,
$$
or, see e.g.~\cite{Vol05}, in an equivalent form
$$
\Fc = i e^{-\pi i (p_1^2+p_2^2)} e^{-\pi i (x_1^2+x_2^2)} e^{-\pi i (p_1^2+p_2^2)}.
$$ 

In fact, choosing the following lifts of the projective transformations $\mathbb{I}_{\sigma_+^{-1}}, \mathbb{I}_{\sigma_-}$ to linear ones defines a lift of the projective represention $\mathbb{I}$ to a unitary representation $I$  of the same group $SL(2,\mathbb{Z})$:
\beq
\label{eq:dehn-action}
\begin{aligned}
{I}_{\sigma_+^{-1}} &= \zeta_se^{\frac23\pi i \tau^2}e^{-\pi i(p_1^2+p_2^2)} \varphi(p_2-p_1+x_2-x_1)^{-1},\\
{I}_{\sigma_-} &= \zeta_s^{-1}e^{-\frac43\pi i \tau^2}e^{2\pi i \tau(x_1-x_2)} e^{-\pi i(x_1^2+x_2^2)} \varphi(p_1-p_2+2\tau)^{-1},
\end{aligned}
\eeq
where we have set
\beq
\label{eq:zetas}
\zeta_s = e^{{\frac{\pi i}{6}(1-c_b^2) }}.
\eeq
Most important for us will be the unitary automorphism $I_S$ of $L^2(\R^2,d\bs x)$: 
\beq
\label{eq:S}
%I_S = \zeta_s I_{\sigma_+^{-1}}I_{\sigma_-}I_{\sigma_+^{-1}}
I_S = I_{\sigma_+^{-1}}I_{\sigma_-}I_{\sigma_+^{-1}}.
\eeq
%with
%\beq
%\label{eq:zetas}
%\zeta_s = e^{{\frac{\pi i}{6}(1-c_b^2) }}.

\begin{prop}
\label{prop:S2}
The element $I_S^2$ acts on $L^2(\R^2,d\bs x)$ by
$$
(I_S^2f)(\bs x) = f(\bs{x}^*) \qquad\text{where}\qquad \bs x^* = (-x_2,-x_1).
$$
In particular, $I_S$ is a unitary automorphism of $L^2(\R^2,d\bs x)$ of order 4.
\end{prop}

\begin{proof}
The proof is a straightforward calculation, based on the inversion formula~\eqref{inv}, the factorization
$$
\Fc = i e^{-\pi i (p_1^2+p_2^2)} e^{-\pi i (x_1^2+x_2^2)} e^{-\pi i (p_1^2+p_2^2)},
$$
and the relation
$$
(e^{2\pi i p_1p_2}e^{2\pi i x_1x_2}e^{2\pi i p_1p_2}f)(\bs x) = (\Fc f)(\bs x^*).
$$
\end{proof}
%
%\begin{remark}
%Using the standard presentation 
%$$
%SL(2,\mathbb{Z}) = \ha{S,T \,\big|\, S^4=1, (ST)^3=S^2},
%$$
%it is easy to see that the assignment
%$$
%S \longmapsto I_S, \qquad T \longmapsto I_T = \zeta_s^{-1}e^{-\frac23\pi i \tau^2} I_{\sigma_+}
%$$
%lifts the projective representation $\mathbb{I}$ to a linear representation of $SL(2,\Z)$.
%\end{remark}
%
%\red{To fix the normalization, put the $\zeta_s$ into definition of $I_{}$}

%In what follows, we will abuse notation and write $\sigma_\pm$ for the projective transformations $I_{\sigma_{\pm}}$ by which they act on $L_2(\mathbb{R}^2)$ under the identifications above.

%To match with~\cite{BJ18}, we set
%\begin{align*}
%Q_1 &= L_{(1,0)} & Q_2 &= \Delta_{(1,0)}, & R = qL_{(0,1)}, \\
%P_1 &= L_{(-1,1)} & P_2 &= \Delta_{(-1,1)},
%\end{align*}
%with
%$$
%t^2 = -qY_{e_2+e_3+e_4}^{\pm1}.
%$$

\subsection{Cluster realization of spherical double affine Hecke algebra.} 
We now algebraically interpret the computations from the previous section as defining an $SL(2;\mathbb{Z})$-equivariant injective homomorphism from the spherical subalgebra $\SH_{q,t}$ of the $\gl_2$ double affine Hecke algebra $\Hbb_{q,t}$ into the universally Laurent algebra $\Lbb^q_{\bs Q}$. The $\gl_n$ version of this homomorphism will be given in a forthcoming publication.

\begin{theorem}
There is an $SL(2,\Z)$-equivariant injective homomorphism
$$
\iota \colon \SH_{q,t} \hookrightarrow \Lbb^q_{\bs Q},
$$
defined by
$$
\iota\hr{E_{(\pm1,0)}} = L_{(\pm1,0)}
\qquad\text{and}\qquad
\iota\hr{E_{(0,\pm1)}} = L_{(0,\pm1)}.
$$
\end{theorem}

\begin{proof}
Comparing formulas~\eqref{eigen} and~\eqref{eq:macdo-under-whit} with the ones ~\eqref{eq:L-action} derived in the previous section shows that for each $(a,b)\in\{(\pm 1,0),(0,\pm1)\}$, the inverse Whittaker transform 
$$
\Wc^* \colon L^2_{\mathrm{sym}}(\R^2, m(\bs\la)d\bs\la) \longrightarrow L^2(\R^2, d\bs x)
$$ intertwines the action of the spherical DAHA generator $E_{(a,b)}$ in the Cherednik representation with that of the quantum cluster algebra element $L_{(a,b)}$ in the representation $\mathcal{S}_\tau\subset\mathcal{H}_\tau$ of the quantum cluster variety. Hence Theorem~\ref{thm:Whittaker}, together with the faithfulness of the two representations at hand, guarantees that $\iota$ is an injective homomorphism. The $SL(2,\Z)$-equivariance is manifest, since for any $g \in SL(2,\mathbb{Z})$ and any primitive vector $v \in \Z^{ 2}$ we have $L_{gv} = gL_v$ and $E_{gv} = gE_v$.
\end{proof}
We therefore see that in the context of quantum Teichm\"uller theory for the punctured torus as formulated in Section~\ref{sec:cluster-torus}, the Whittaker functions $\Psi_{\bs\la}(\bs x)$ are interpreted as eigenfunctions for the quantized traces of monodromy around the $(1,0)$-curve. Equation ~\eqref{10to01} is interpreted as the intertwining relation
\beq
\label{Toda-Macdo}
\Mtau_j \Wc I_S = \Wc I_S H_j.
\eeq 
between the Toda and Macdonald operators under the unitary equivalence 
$$
\Wc \circ I_S \colon L^2(\R^2, d\bs x) \longrightarrow L^2_{\mathrm{sym}}(\R^2, m(\bs\la)d\bs\la).
$$
In particular, the distributions $I_S \Psi_{\bs\la}(\bs x)$ provide a complete set of eigenfunctions for the quantized traces of monodromy around the $(0,1)$-curve.

\begin{remark}
\label{rmk:dehn}
Comparing formulas~\eqref{eq:qsys} and~\eqref{eq:dehn-actions}, we see that the $Q$-system evolution operator is identified with the action $I_{\sigma_+}$ of the $(1,0)$-cycle Dehn twist, justifying our notation for the former.
\end{remark}

\section{Ruijsenaars wavefunctions as matrix coefficients.}
\label{sec:wave}

In this section we show how the eigenfunctions of Macdonald operators, the Halln\"as--Ruijsenaars functions, can be derived from those of Toda Hamiltonians using the action of $SL(2;\mathbb{Z})$ on the Hilbert space assigned by quantum Teichm\"uller theory to the punctured torus. More specifically, we present a formula for the Halln\"as--Ruijsenaars eigenfunctions as matrix coefficients between Toda eigenfunctions for the element of $S$ of $SL(2,\mathbb{Z})$ considered in the previous section. 

Recall from the previous section the eigenfunctions $I_S\cdot \Psi_{\bs\la}(\bs x)$ for the quantized traces of monodromy around the $(0,1)$-curve. We first compute these eigenfunctions more explicitly. 

\begin{prop}
\label{prop:S-action-on-Psi}
The action of the unitary transformations $I_S$ and $I_S^2$ on the distribution $\Psi_{\bs\la}(\bs x)$ read
\begin{align*}
\label{eq:S-action-on-Psi}
I_S \Psi_{\bs\la}(\bs x) &= \zeta \zeta_s e^{-\pi i (\la_1^2+\la_2^2)} e^{\pi i c_b(2\tau-\underline{\bs\la})} e^{2\pi i (x_1+c_b)x_2} \delta(\underline{\bs\la}-\underline{\bs x}) \\
&\phantom{{}={}}\varphi(x_1-\mu_2-\tau+c_b) \varphi(\la_2-x_2-\tau+c_b), \\
I_S^2 \Psi_{\bs\la}(\bs x)  &= \Psi_{\bs\la^*}(\bs x),
\end{align*}
where we recall from Appendix~\ref{sec:qdl} the phase constant $\zeta = e^{\pi i(1-4c_b^2)/12}$.
\end{prop}
\begin{proof}
First we observe that since $\Psi_{\bs\mu}(\bs x) $ is an eigenfunction of the Dehn twist $\sigma_+$ with eigenvalue given by~\eqref{eq:gamma-eigen}, we have
$$
I_S \Psi_{\bs\mu}(\bs x) = \zeta_s e^{-\pi i (\mu_1^2+\mu_2^2)} I_{\sigma} \Psi_{\bs\mu}(\bs x).
$$
So it remains to calculate $I_\sigma \Psi_{\bs\mu}$, which is given by the integral
$$
I_\sigma \Psi_{\bs\mu}(\bs x) = e^{2\pi ix_1x_2} \int  e^{-2 \pi i(x_1y_1+x_2y_2)} e^{2\pi i \tau(y_1-y_2)}
\Psi_{\bs\mu}(\bs y) d\bs y.
$$
Using the Gauss--Givental formula~\eqref{GG} and substituting $y_1 = s-z_1$, $y_2 = z_2+s-c_b$, we arrive at
$$
\zeta^{-1} e^{-\pi i c_b \underline{\bs\mu}} e^{2\pi i c_b(x_2+\tau)} \int \frac{e^{2\pi i z_1(x_1-\mu_2-\tau)} e^{2\pi i z_2(\mu_2-x_2-\tau)} e^{2\pi i s(\underline{\bs\mu}-\underline{\bs x})}} {\varphi(z_1-c_b)\varphi(z_2-c_b)} ds d\bs z.
$$
To finish the proof of the first formula, we take the three integrals in the above expression by applying the standard formula for the delta-function
$$
\int e^{2 \pi i s x} ds = \delta(x)
$$
along with the quantum dilogarithm Fourier transform~\eqref{Fourier-1}.

In view of the Proposition~\ref{prop:S2}, the second formula for the action of $I_S^2$ follows from comparing the Gauss--Givental presentation~\eqref{GG} for $\Psi_{\bs\la^*}(\bs x)$ with the Mellin--Barnes presentation~\eqref{MB} for $\Psi_{\bs\la}(\bs x^*)$.
%by shifting the integration variable $r\mapsto r-\underline{\bs x}$ in the Gauss-Givental formula~\eqref{GG} for $\Psi_{\bs\mu}$.
\end{proof}

We denote by $\Phi_{\bs\mu}^\tau(\bs\la)$ the following matrix coefficient
\beq
\label{eq:mc}
\Phi_{\bs\mu}^\tau(\bs\la) = \ha{\Psi_{\bs\la}, I_S \Psi_{\bs\mu}}
\eeq
of the element $S \in SL(2,\Z)$.

\begin{theorem}
\label{theorem2}
Let $\Mtau_j$ be the $j$-th $\gl_2$-Macdonald operator.
\begin{enumerate}
\item[(1)] The Macdonald operators $\Mtau_1$, $\Mtau_2$ are essentially self-adjoint with respect to $\ha{\cdot,\cdot}_S$.
\item[(2)] The function $\Phi_{\bs\mu}^\tau(\bs\la)$ is an eigenfunction of Macdonald operators: for $j=1,2$ one has
$$
\Mtau_j \Phi_{\bs\mu}^\tau(\bs\la) = e_j(\bs\mu) \Phi_{\bs\mu}^\tau(\bs\la).
$$
\item [(3)]For any $\tau\in\R$, the integral transform
$$
\Mc \colon f \longmapsto \ha{\Phi_{\bs\mu}^\tau,f}
%\Mc \colon f(\bs\la) \longmapsto \int_{\R^2} f(\bs\la) \overline{\Phi_{\bs\mu}^\tau(\bs\la)} m(\bs\la) d\bs\la
$$
is a unitary automorphism of $L^2_{\mathrm{sym}}(\R^n, m(\bs\la)d\bs\la)$.
\end{enumerate}
\end{theorem}

\begin{proof}
Points~1 and~3 follow from the equality~\eqref{Toda-Macdo}, the unitarity of $I_S$, Theorem~\ref{thm:Whittaker}, and the fact that Toda Hamiltonians are essentially self-adjoint with respect to $\ha{\cdot,\cdot}$. Point~2 is a consequence of the equalities~\eqref{eigen}, \eqref{Toda-Macdo}. 
\end{proof}

\begin{cor}
\label{cor:Macdo}
The matrix coefficient $\Phi_{\bs\mu}^\tau(\bs\la)$ has the following properties:
\begin{enumerate}
\item [(1)] The function $\Phi_{\bs\mu}^\tau(\bs\la)$ is symmetric in $\bs \lambda$, $\bs\mu$, and satisfies
\beq
\label{eq:Phi-unitarity}
\overline{\Phi_{\bs\mu}^\tau(\bs\la)} = \Phi_{\bs\mu}^{-\tau}(\bs\la^*).
\eeq
\item [(2)] The function $\Phi_{\bs\mu}^\tau(\bs\la)$ satisfies the duality
$$
\Phi_{\bs\mu}^\tau(\bs\la) = \Phi_{\bs\la}^{-\tau}(\bs\mu).
$$
\end{enumerate}
\end{cor}

\begin{proof}
By Point 3 of the Theorem~\ref{theorem2}, the symmetry in $\bs\la$ and $\bs\mu$ is inherited from that of the Whittaker functions, and the equality~\eqref{eq:Phi-unitarity} is easily checked using the unitarity property~\eqref{inv} of the non-compact quantum dilogarithm. For Point 2, we use the symmetries in Point 1, together with the unitarity of the $S$-transformation $I_S$ and the second part of Proposition~\ref{prop:S-action-on-Psi}:
$$
\Phi_{\bs\mu}^\tau(\bs\la)
= \langle \Psi_{\bs\la},I_S \Psi_{\bs\mu} \rangle
%= \langle I_S^3 \Psi_{\bs\la},\Psi_{\bs\mu} \rangle
= \langle I_S\Psi_{\bs\la}, I_S^2\Psi_{\bs\mu} \rangle
= \langle I_S\Psi_{\bs\la}, \Psi_{\bs\mu^*} \rangle
= \overline{ \langle\Psi_{\bs\mu^*}, I_S\Psi_{\bs\la} \rangle}
= \overline{ \Phi_{\bs\la}^\tau(\bs\mu^*)}
= \Phi_{\bs\la}^{-\tau}(\bs\mu).
$$
\end{proof}

The following Proposition gives an explicit calculation of the matrix element~\ref{eq:mc}. An equivalent formula has been obtained by Teschner and Vartanov in Section 6.5.4 of~\cite{TV15}.
\begin{prop}
\label{prop:RH}
The Halln\"as--Ruijsenaars eigenfunction $\Rho_{\bs\mu}^\tau(\bs\la)$ can be expressed as
\beq
\label{Macdo-mc}
\Rho_{\bs\mu}^\tau(\bs\la) = \zeta \zeta_s^{-1} e^{2\pi i\tau^2} \Phi_{\bs\mu}^\tau(\bs\la).
\eeq
\end{prop}

\begin{proof}
Applying~\eqref{eq:qdl-unitary} and~\eqref{inv} to the Gauss--Givental presentation~\eqref{GG} of the Whittaker function, we get
$$
\overline{\Psi_{\bs\la}(\bs x)} = \zeta^{-3} e^{-2\pi ic_b^2} e^{\pi i c_b(\la_2-\la_1)} e^{-2\pi i \la_2 \underline{\bs x}} \int \frac{e^{\pi i (r-x_1+c_b)^2} e^{\pi i (x_2-r)^2} e^{2\pi i r(\la_2-\la_1)} }{\varphi(x_1-r-c_b) \varphi(r-x_2)} dr.
$$
We now substitute the above expression into~\eqref{eq:mc} and apply Proposition~\ref{prop:S-action-on-Psi}. Integrating out $x_2$ and shifting the integration variable $x_1 = x+r$, we have
\begin{align*}
\ha{\Psi_{\bs\la}, I_S \Psi_{\bs\mu}}
 = \zeta_{\mathrm{inv}} \zeta_s e^{\pi i c_b(\underline{\bs\mu}+\la_2-\la_1+2\tau)} e^{2\pi i \mu_1\mu_2} e^{-2 \pi i \la_2\underline{\bs\mu}} \int e^{2\pi i r(r+\la_2-\la_1-\underline{\bs\mu}-c_b)} \\
e^{-4\pi i c_bx} \frac{\varphi(x+r-\mu_1-\tau+c_b) \varphi(x+r-\mu_2-\tau+c_b)}{\varphi(x+2r-\underline{\bs\mu}) \varphi(x-c_b)} dx dr
\end{align*}
Using the identity~\eqref{Saal} with $u_j = r-\mu_j-\tau+c_b$ and $v = 2r-\underline{\bs\mu}+c_b$ and applying the inversion formula~\eqref{inv} three more times, we derive that the integral over $x$ is equal to
$$
\zeta e^{-2\pi i(\mu_1\mu_2+\underline{\bs\mu}\tau +\tau^2)} \varphi(2\tau) e^{2\pi i r(-r+\underline{\bs\mu}+2\tau)} \frac{\varphi(r-\mu_1-\tau+c_b) \varphi(r-\mu_2-\tau+c_b)}{\varphi(r-\mu_1+\tau) \varphi(r-\mu_1-\tau)}.
$$
Finally, substituting $r = x + \frac12(\underline{\bs\mu}-c_b)$ we arrive at the desired formula~\eqref{Macdo-mc}.
\end{proof}

The following result is an immediate consequence of Corollary~\ref{cor:Macdo} and Proposition~\ref{prop:RH}.

\begin{cor}
\label{cor:cor}
The Halln\"as--Ruijsenaars eigenfunction $\Rho_{\bs\mu}^\tau(\bs\la)$ is symmetric in $\bs \lambda$, $\bs\mu$, and satisfies
\begin{align*}
\Rho_{\bs\mu}^\tau(\bs\la) = \Rho_{\bs\la}^{-\tau}(\bs\mu)
\qquad\text{and}\qquad
\overline{\Rho_{\bs\mu}^\tau(\bs\la)} = \zeta_{inv}^{-1}e^{-4\pi i \tau^2}{\Rho_{\bs\mu}^{-\tau}(\bs\la^*)}.
\end{align*}
\end{cor}

\section{Macdonald polynomials and Harish-Chandra series}
\label{sec:poly}

We finish the paper by showing how the $GL(2)$ $q$-Whittaker and Macdonald polynomials can be obtained respectively as special values of the analytic continuations of Whittaker and Halln\"as--Ruijsenaars functions. We also show that by closing the contour of integration and calculating residues of these functions, we recover the Harish-Chandra series in both Whittaker and Macdonald case.

\subsection{Whittaker polynomials from Whittaker functions}

As was observed in~\cite{SS18}, the Whittaker function $\Psi_{\bs\la}(\bs x)$ becomes entire after multiplying by $\varphi(x_2-x_1)$. Let us define
\beq
\label{eq:psitilde}
\widetilde\Psi_{\bs\la}(\bs x) = -\zeta^{-1} e^{\pi i c_b (\la_2-\la_1)} e^{2\pi i (\la_1 x_1 + \la_2 x_2)} \varphi(x_2-x_1) \int_C \frac{e^{2\pi i t(\la_1-\la_2) }} {\varphi(t-c_b)\varphi(x_2-x_1-t)}dt,
\eeq
and note that $\widetilde\Psi_{\bs\la}(\bs x)$ is a joint eigenfunction of the mutated Hamiltonians
\begin{align*}
\varphi(x_2-x_1) H_1 \varphi(x_2-x_1)^{-1} &= e^{2\pi b p_1} + e^{2\pi b(p_1+x_2-x_1)} + e^{2\pi b p_2}, \\
\varphi(x_2-x_1) H_2 \varphi(x_2-x_1)^{-1} &= e^{2\pi b(p_1+p_2)}.
\end{align*}

For the remainder of this section we let $\bs n, \bs{\tilde n} \in \Z^2$ be a pair of generalized partitions, each with 2 parts. That is
\beq
\label{not:n}
\begin{aligned}
\bs n &= (n_1,n_2), \qquad n_1 \le n_2,\\
\bs{\tilde n} &= (\tilde n_1, \tilde n_2), \qquad \tilde n_1 \le \tilde n_2.
\end{aligned}
\eeq
We also fix a notation
\beq
\label{not:crs}
c^\pm_{r,s} = \pm\frac{1}{2}c_b -irb-isb^{-1}.
\eeq

%Note that both $\Psi_{\bs\la}(\bs x)$ and $\widetilde\Psi_{\bs\la}(\bs x)$ only depend only on the difference $x_2-x_1$ when $\la_1+\la_2=0$. Setting $x = x_2-x_1$ we define the $\sl_2$ Whittaker function by
%$$
%\widetilde\Psi_{\la}(x) = \widetilde\Psi_{-\la,\la}(x_1,x_2).
%$$

\begin{defn}
The $\gl_2$ Whittaker polynomial $W_{\bs n}(\bs z;q)$ is defined by
\beq
\label{eq:qGLWhit}
W_{\bs n}(\bs z; q) = \sum_{k=0}^{n_2-n_1}\binom{n_2-n_1}{k}_q z_1^{n_1+k}z_2^{n_2-k},
\eeq
with the $q$-binomial coefficient given by~\eqref{q-binom}.
%with the $q$-binomial coefficient being
%$$
%\binom{n}{k}_q = \frac{(q;q)_n}{(q;q)_k(q;q)_{n-k}}.
%$$
\end{defn}

\begin{theorem}
\label{thm:W-pol}
The value of the function $\widetilde\Psi_{\bs\la}(\bs x)$ at the point $\bs x = (c^+_{n_1,\tilde n_1},c^-_{n_2,\tilde n_2})$ is given by
\beq
\label{eq:whittaker-analytic-continuation}
\widetilde\Psi_{\bs\la}(c^+_{n_1,\tilde n_1},c^-_{n_2,\tilde n_2}) = W_{\bs n}(e^{2\pi b \la_1},e^{2\pi b \la_2};q^{-2}) W_{\bs {\tilde n}}(e^{2\pi b^{-1} \la_1},e^{2\pi b^{-1} \la_2};\tilde q^{-2}).
\eeq
\end{theorem}

\begin{proof}
Let us first consider the value of $\widetilde\Psi_{\la}(\bs x)$ at the point
$$
\bs x = (c^+_{n_1,\tilde n_1},c^-_{n_2,\tilde n_2}-\epsilon)
$$
for some $\epsilon\in\R$. Recalling the poles and zeros of the non-compact dilogarithm~\eqref{poles-zeros}, we see that the downward and upward running sequence of poles of the integrand consist respectively of the points
$$
t^-_{l,\tilde l} = -ilb-i\tilde l b^{-1} \qquad\text{and}\qquad t^+_{m,\tilde m} = \epsilon +ib(m-n)+ib^{-1}(\tilde m-\tilde n)
$$
where
$$
n = n_2-n_1, \qquad \tilde n = \tilde n_2 - \tilde n_1, \qquad\text{and}\qquad l, \tilde l, m, \tilde m \in\Z_{\ge0}.
$$
Recall now that the integration contour $C$ in~\eqref{GG} is defined so as to cut the plane into two connected components, one containing the sequence $\{t^-_{l,\tilde l}\}$ and the other the sequence  $\{t^+_{m,\tilde m}\}$.

As $\epsilon\rightarrow 0$, two phenomena simultaneously occur: the prefactor $\varphi(x_2-x_1)$ tends to zero, while the $(n+1)(\tilde n+1)$ poles of the integrand  with
$$
l+m=n
\qquad\text{and}\qquad
\tilde l+\tilde m=\tilde n
$$
collide. Let us write $I$ for the integral obtained by pushing the contour of integration $C$ across the colliding poles $t^+_{m,\tilde m}$, so that setting $\widetilde\Psi= \int_C\tilde\psi$, we have
$$
\widetilde \Psi = I + 2\pi i\sum_{m=0}^n\sum_{\tilde m=0}^{\tilde n}\mathrm{Res}_{t=t^+_{m,\tilde m}}\tilde\psi.
$$
We also have $\lim_{\epsilon \to 0}I=0$ owing to the vanishing of $\varphi(x_2-x_1)$. Finally, using~\eqref{eqn:res} and~\eqref{eqn:func} to calculate the sum of $(n+1)(\tilde n+1)$ residues, we arrive at the equality~\eqref{eq:whittaker-analytic-continuation}. 
\end{proof}

\subsection{Macdonald polynomials from Halln\"as--Ruijsenaars functions} The derivation of Macdonald polynomials goes along the similar lines as that of Whittaker polynomials.

\begin{defn}
Given a partition $\bs n$ as in~\eqref{not:n}, the symmetric $\gl_2$ Macdonald polynomial is defined by
$$
%\label{Macdo-poly}
P_{\bs n}(\bs x; t,q) = \sum_{r=0}^{n_2-n_1} \frac{\hr{q^{n_2-n_1};q^{-1}}_r \hr{t;q}_r}{\hr{q^{n_2-n_1-1}t;q^{-1}}_r \hr{q;q}_r} x_1^{n_1+r} x_2^{n_2-r},
$$
where $(X;q)_n$ denotes the $q$-Pochhammer symbol, see~\eqref{q-Poch}.
\end{defn}

Consider the renormalized Halln\"as--Ruijsenaars function
\beq
\label{eq:HR-renorm}
\widetilde \Phi^\tau_{\bs\mu}(\bs\la) = \zeta e^{-4\pi i \mu\tau_-} \frac{\varphi(-2\mu-c_b)}{\varphi(-2\mu-2\tau)} \Rho^\tau_{\bs\mu}(\bs\la).
\eeq

\begin{theorem}
\label{thm:analytic-continuation}
 Given a pair of partitions $\bs n$, $\bs{\tilde n}$ as in~\eqref{not:n}, the value of the function $\widetilde \Phi^\tau_{\bs\mu}(\bs\la)$ at the point $\bs\mu = (-\tau-c^+_{n_1,\tilde n_1},\tau-c^-_{n_2,\tilde n_2})$ is given by the product
\beq
\label{eq:prod-Macdo}
P_{\bs n}(e^{2\pi b\la_1}, e^{2\pi b\la_2}; e^{2\pi b(2\tau+c_b)}, e^{2\pi i b^2}) P_{\bs{\tilde n}}(e^{2\pi b^{-1}\la_1}, e^{2\pi b^{-1}\la_2}; e^{2\pi b^{-1}(2\tau+c_b)}, e^{2\pi i b^{-2}}),
\eeq
where $c^{\pm}_{r,s}$ is defined by~\eqref{not:crs}.
\end{theorem}

\begin{proof}
We show the conclusion of the theorem holds for $\tau_-\neq0$, with the $\tau_-=0$ case then following from the analyticity of both sides in $\tau$. As in the proof of Theorem~\ref{thm:W-pol}, we consider the value of $\widetilde \Phi^\tau_{\bs\mu}(\bs\la)$ at
$$
\bs\mu = (-\tau-c^+_{n_1,\tilde n_1}-\eps,\tau-c^-_{n_2,\tilde n_2}+\eps).
$$
As $\eps \to 0$, we see that the zeros
$$
x^+_{m,\tilde m} = ib\hr{\frac{n_2-n_1}{2}-m} + ib\hr{\frac{\tilde n_2 - \tilde n_1}{2}-\tilde m} + \eps, \qquad m, \tilde m \in \Z_{\ge0}
$$
of the dilogarithm $\phi(x-\mu-\tau_-)$ collide with the poles
$$
x^-_{l,\tilde l} = ib\hr{l-\frac{n_2-n_1}{2}} + ib\hr{\tilde l - \frac{\tilde n_2 - \tilde n_1}{2}} - \eps, \qquad l, \tilde l \in \Z_{\ge0}
$$
of the dilogarithm $\phi(x+\mu+\tau_-)$. Following the reasoning in the proof of Theorem~\ref{thm:W-pol}, we obtain
$$
\widetilde \Phi = -2\pi i\sum_{m=0}^{n_2-n_1} \sum_{\tilde m=0}^{\tilde n_2 - \tilde n_1}\mathrm{Res}_{x=x^+_{m,\tilde m}}\tilde\phi,
$$
where $\widetilde \Phi = \int_C \tilde \phi$. Finally, the evaluation of residues yields the desired result.
\end{proof}

\subsection{Harish-Chandra series from Whittaker functions}

By the convention~\eqref{contour-convention}, the contour $C$ in the integral~\eqref{eq:psitilde} passes above the zeros of $\varphi(t-c_b)$, below those of $\varphi(x_2-x_1-t)$, and can be chosen to escape to infinity along the real line. The zeros of the factor $\varphi(t-c_b)$ read $t_{r,s}=-i(br+b^{-1}s)$, $r,s\in \Z_{\ge0}$. Using equations (A.4) and (A.6), we find that
$$
\Res_{t=t_{r,s}} \hr{e^{2\pi it(\lambda_1-\lambda_2)}\frac{\varphi(x_2-x_1)}{\varphi(t-c_b)\varphi(x_2-x_1-t)}} = -\zeta \Lambda_{1,2}^r \widetilde \Lambda_{1,2}^s \frac{(-qX_{2,1};q^2)_r}{(q^{-2r};q^{2})_r} \frac{(-q\widetilde X_{2,1};\tilde q^2)_s}{(\tilde q^{-2r};\tilde q^{2})_s}
$$
where $X_j=e^{2\pi bx_j}$, $\widetilde X_j=e^{2\pi b^{-1} x_j}$, and we write $A_{j,k}$ for the ratio $A_j/A_k$. Making a change of variables
$$
x_1'=x_1-\frac{c_b}{2}, \qquad x_2'=x_2+\frac{c_b}{2}
$$
and setting $X_j'=e^{2\pi b x_j'}$, we arrive at the (asymptotically free) Harish-Chandra series for the $\gl_2$ Whittaker $q$-difference equation, see e.g.~\cite{Che09, DFKT17},\footnote{We would like to note that in \emph{loc.cit.} variables $\La$ and $X$ are interchanged.}
%$$
%\widetilde \Psi_{\bs\la}(\bs x)=e^{2\pi i(\lambda_1x_1'+\lambda_2 x_2')} W_q(\La_{1,2}\vert X'_{2,1}) W_{\tilde q}(\widetilde\Lambda_{1,2} \vert \widetilde X'_{2,1}),
%$$
\beq
\label{eq:whit-series}
\sum_{r \in \Z_{\ge0}}\Res_{t=t_{r,0}}(\tilde \psi) = \La_1^{ib^{-1}x'_1} \La_2^{ib^{-1}x'_2} \sum_{r\ge 0} \Lambda_{1,2}^r  \frac{(X'_{2,1};q^2)_r}{(q^{-2r};q^2)_r}.
\eeq
One can also check that a similar calculation with the zeros of the factor $\varphi(x_2-x_1-t)$ yields an expression~\eqref{eq:whit-series} with $\La_1$ and $\La_2$ interchanged.
%where
%$$
%W_q(\Lambda\vert X)=\sum_{r\ge 0} \Lambda^r  \frac{(X;q^2)_r}{(q^{-2r};q^2)_r}.
%$$
%%We now note that the exponential prefactor may be expressed in two ways:
%%$$
%%e^{2i\pi(\lambda_1x_1'+\lambda_2 x_2')}= \Lambda_1^{i b^{-1} x_1'}  \Lambda_2^{i b^{-1} x_2'} = \widetilde \Lambda_1^{i b x_1'} \widetilde \Lambda_2^{i b x_2'}.
%%$$
%Setting
%$$
%x_j''=i b^{-1} x_j', \qquad \tilde x_j''=i bx_j'
%$$
%for $j=1,2$, and introducing notation
%$$
%\bs\La^{\bs x} = \La_1^{x_1} \La_2^{x_2},
%$$
%we arrive at the following factorized expressions:
%\beq
%\label{eq:whit-series}
%\begin{aligned}
%\widetilde \Psi_{\bs\lambda}(\bs x) &= \bs\Lambda^{\bs{x''}} W_q(\Lambda_{1,2}\vert X'_{2,1}) W_{\tilde q}(\widetilde \Lambda_{1,2}\vert \widetilde X'_{2,1}) \\ 
%&= \bs{\widetilde \Lambda}^{\bs{\tilde x''}} W_q(\Lambda_{1,2}\vert X'_{2,1}) W_{\tilde q}( \widetilde \Lambda_{1,2}\vert \widetilde X'_{2,1}),
%\end{aligned}
%\eeq
%where $\bs\Lambda^{\bs{x''}} W_q(\Lambda_{1,2}\vert X'_{2,1})$ is the (asymptotically free) Harish-Chandra series for the $\gl_2$ Whittaker $q$-difference equation, see e.g.~\cite{Che09, DFKT17}.
Finally, we note that conditions $x_1=c^+_{n_1,\tilde n_1}$, $x_2=c^-_{n_2,\tilde n_2}$ imply $X_{2,1}'=q^{2(n_1-n_2)}$,
% \qquad \widetilde X_{2,1}'=\tilde q^{2(n_1-n_2)},
and the Harish-Chandra series~\eqref{eq:whit-series} truncates to the first factor in~\eqref{eq:whittaker-analytic-continuation}.

\subsection{Harish-Chandra series from Halln\"as--Ruijsenaars functions}
Similarly to the Whittaker case, we shall calculate the sum of the residues of the integrand in the expression~\eqref{eq:HR} at points
$$
x_{r,0}^{\pm} = \pm\mu-\tau + c_{r,0}^-,
$$
where $\mu = (\mu_1-\mu_2)/2$ and $c_{r,s}^-$ is defined by~\eqref{not:crs}. Let us denote
$$
\mu'_1 = \mu_1-\tau_+, \qquad \mu'_2 = \mu_2+\tau_+,
\qquad\text{and}\qquad
\mu''_j = -ib^{-1}\mu'_j
$$
for $j=1,2$. Then the contribution $\Rho_+$ from the poles at $x=x_{r,0}^+$ reads:
$$
\Rho_+ = \zeta\zeta_{inv} e^{4\pi i \tau^2} e^{2\pi i \tau_-(\mu'_2-\mu'_1)} \frac{\varphi(\mu_1-\mu_2-2\tau)}{\varphi(\mu_1-\mu_2-c_b)} {\bs\La}^{\bs\mu''} P_{q,t}(\La_{2,1}\vert \Mu_{1,2}), %P_{\tilde q, \tilde t}(\widetilde\La_{2,1}\vert \widetilde \Mu_{1,2}),
$$
where
$$
\bs\La^{\bs \mu''} = \La_1^{\mu''_1} \La_2^{\bs\mu''_2}
\qquad\text{and}\qquad
P_{q,t}(\La\vert\Mu)  = \sum_{r \ge 0} \La^r \frac{(t^{-2}\Mu;q^{-2})_r (t^2;q^2)_r}{(q^{-2}\Mu;q^{-2})_r (q^2;q^2)_r}.
$$
We note that the product ${\bs\La}^{\bs{\mu''}} P_{q,t}(\La_{2,1}\vert \Mu_{1,2})$ is the Harish-Chandra series solution to the Macdonald eigenvalue equation, see~\cite{Che09,Sto14,NS12}. It is immediate to see that the contribution $\Rho_-$ from the poles at $x=x_{r,0}^-$ can be obtained from $\Rho_+$ by swapping $\mu_1$ and $\mu_2$.

Finally, evaluating the expression
\beq
\label{eq:HR-res}
\zeta e^{-4\pi i \mu\tau_-} \frac{\varphi(-2\mu-c_b)}{\varphi(-2\mu-2\tau)} \hr{\Rho_+ + \Rho_-},
\eeq
at point $\bs\mu = (-\tau-c^+_{n_1,\tilde n_1},\tau-c^-_{n_2,\tilde n_2})$, we see that the first summand vanishes, since the dilogarithm $\varphi(-2\mu-2\tau)$ has a pole, while the second summand truncates to the first factor in~\eqref{eq:prod-Macdo}.

\appendix

\section{Quantum dilogarithms}
\label{sec:qdl}
In this appendix we recall the definition and some properties of the quantum dilogarithm, mostly following~\cite{Kas01} and~\cite{FG09}. 

\subsection{Non-compact quantum dilogarithm}

\begin{defn}
The \emph{non-compact quantum dilogarithm function} $\varphi_b(z)$ is defined in the strip $\hm{\Im(z)} < \hm{\Im(c_b)}$ by the following formula:
$$
\varphi_b(z) = \exp\hr{\frac{1}{4} \int_C\frac{e^{-2izt}}{\sh(tb)\sh(tb^{-1})}\frac{dt}{t}},
$$
where the contour $C$ follows the real line from $-\infty$ to $+\infty$, surpassing the origin in a small semi-circle from above. 
\end{defn}

The non-compact quantum dilogarithm can be analytically continued to the complex plane as a meromorphic function with an essential singularity at infinity. The resulting function, which we denote by the same symbol $\varphi_b(z)$, enjoys the symmetry
$$
\varphi_b(z) = \varphi_{-b}(z) = \varphi_{b^{-1}}(z).
$$
In what follows, we shall write $\varphi(z)$ for $\varphi_b(z)$, omitting the symbol $b$ from the notation. If $b+b^{-1}\in\mathbb{R}$, then the function $\varphi$ satisfies the unitarity relation 
\begin{align}
\label{eq:qdl-unitary}
\overline{\varphi(z)} \varphi(\overline{z}) = 1.
\end{align}
We assume $b+b^{-1} \in \R_{>0}$, and recall the pure imaginary constant
$$
c_b = \frac{i(b + b^{-1})}{2}
$$
along with phase constants
\beq
\label{eq:zetas}
\zeta = e^{\pi i(1-4c_b^2)/12} \qquad\text{and}\qquad \zeta_{\inv} = \zeta^{-2} e^{-\pi i c_b^2}.
\eeq
Then the poles and zeros of $\varphi(z)$ are given by
\beq
\label{poles-zeros}
\varphi(z)^{\pm1} = 0 \quad\Leftrightarrow\quad z = \mp\hr{c_b + ibm + ib^{-1}n} \quad\text{for}\quad m,n \in \Z_{\ge 0}.
\eeq
The asymptotic behavior around its zeros, poles, and infinity is described by
\beq
\label{eqn:res}
\varphi(z\pm c_b) \sim \pm \zeta^{-1} (2 \pi i z)^{\mp1} \qquad\text{as}\qquad z \to 0,
\eeq
and
$$
\varphi(z) \sim
\begin{cases}
\zeta_{\inv} e^{\pi i z^2}, & \Re(z) \to +\infty, \\
1, & \Re(z) \to -\infty;
\end{cases}
$$
respectively. Finally, $\varphi$ satisfies the reflection/inversion relation,
\beq
\label{inv}
\varphi(z) \varphi(-z) = \zeta_{\inv} e^{\pi i z^2},
\eeq
as well as the functional equations
\beq
\label{eqn:func}
\varphi\hr{z - ib^{\pm1}/2} = \hr{1 + e^{2\pi b^{\pm1}z}} \varphi\hr{z + ib^{\pm1}/2},
\eeq
and the pentagon identity
\beq
\varphi(p) \varphi(x) = \varphi(x) \varphi(p+x) \varphi(p),
\label{pentagon}
\eeq
which holds for any pair of self-adjoint operators $p$ and $x$ satisfying $[p,x] = \frac{1}{2\pi i}$.

The non-compact quantum dilogarithm is closely related to the classical Barnes double sine function $S_2(z) = S_2(z \,|\,\omega_1,\omega_2)$ defined by the formula
\beq
\label{eq:barnes}
S_2(z) = \exp\hr{\int_0^\infty\hr{\frac{\sh((2z-\omega_1-\omega_2)t)}{\sh(\omega_1t)\sh(\omega_2t)} - \frac{2z-\omega_1-\omega_2}{\omega_1\omega_2t}} \frac{dt}{2t}}.
\eeq
Indeed, we have
$$
S_2\hr{z \,|\, b,b^{-1}} = \varphi(iz-c_b) e^{-\frac{\pi i}{2}(iz-c_b)^2}.
$$

\subsection{Integral identities.}
The quantum dilogarithm function $\varphi(z)$ satisfies many integral identities, all of which make use of the convention~\eqref{contour-convention}. Namely, in any contour integral of the form 
$$
\int_{C} \prod_{j,k}\frac{\varphi(t-a_j)}{\varphi(t-b_k)}f(t)dt,
$$
where $f(t)$ is some entire function, we assume that the contour $C$ is passing below the poles of $\varphi(t-a_j)$ for all $j$, above the poles of $\varphi(t-b_k)^{-1}$ for all $k$, and escaping to infinity in such a way that the integrand is rapidly decaying.

The Fourier transform of the quantum dilogarithm can be calculated explicitly:
\begin{align}
\label{Fourier-1}
\zeta \varphi(w) &= \int\frac{e^{2\pi i x(w-c_b)}}{\varphi(x-c_b)} dx, \\
\label{Fourier-2}
\frac{1}{\zeta \varphi(w)} &= \int \frac{\varphi(x+c_b)}{e^{2\pi i x(w+c_b)}} dx.
\end{align}
Note that in accordance with Notation~\ref{contour-convention}, the integration contours in~\eqref{Fourier-1} and~\eqref{Fourier-2} can be taken to be $\R+i\eps$ and $\R-i\eps$ respectively with a small positive number $\eps$.

Using the above Fourier transform to compute the integral kernels for the operators $\varphi(p)$ and $\varphi(p+x)$, one sees that the pentagon identity~\eqref{pentagon} is equivalent to either of the following integral analogs of Ramanujan's $_1\psi_1$ summation formula:
%\begin{align}
%\label{beta-1}
%\zeta \frac{\varphi(a) \varphi(w)}{\varphi(a+w-c_b)} &= \int \frac{\varphi(x+a)}{\varphi(x-c_b)} e^{2\pi i x(w-c_b)} dx, \\
%\label{beta-2}
%\zeta^{-1} \frac{\varphi(a+w+c_b)}{\varphi(a) \varphi(w)} &= \int \frac{\varphi(x+c_b)}{\varphi(x+a)} e^{-2\pi i x(w+c_b)} dx.
%\end{align}
%
\begin{align}
\label{beta-1}
\int \frac{\varphi(x+u)}{\varphi(x+v)} e^{2\pi i xw} dx &= \zeta e^{-2\pi i w(v+c_b)} \frac{\varphi(u-v-c_b) \varphi(w+c_b)}{\varphi(u-v+w-c_b)}, \\
\label{beta-2}
&= \zeta^{-1} e^{-2\pi i w(u-c_b)} \frac{\varphi(v-u-w+c_b)}{\varphi(v-u+c_b) \varphi(-w-c_b)}.
\end{align}
Finally, we will use the following limit of the integral analogue of the Saalsch\"utz summation formula:
\begin{align}
\label{Saal}
\int e^{-4\pi i c_b x} \frac{\varphi(x+u_1) \varphi(x+u_2)}{\varphi(x+v-c_b)\varphi(x-c_b)} dx = \zeta^3 e^{\pi i v(2c_b-v)}\frac{\varphi(u_1)\varphi(u_2)\varphi(u_1-v)\varphi(u_2-v)}{\varphi(u_1+u_2-v-c_b)}.
\end{align}

\subsection{Compact quantum dilogarithm}

Recall the standard notation for the $q$-Pochhammer symbol:
\beq
\label{q-Poch}
(X;q)_n = \prod_{k=0}^{n-1}(1-q^kX).
\eeq
so that the $q$-binomial coefficient can be expressed as
\beq
\label{q-binom}
\binom{n}{k}_q = \frac{(q;q)_n}{(q;q)_k(q;q)_{n-k}}.
\eeq
The \emph{compact quantum dilogarithm function} $\Psi_q(X)$ is defined by the series expansion of
\beq
\label{eq:compact-qdl}
\Psi_q(X) = \prod_{n\ge0}\left(1+Xq^{2n+1}\right)^{-1} \in \mathbb{Q}(q)[[X]],
\eeq
or equivalently
$$
\Psi_q(X) = (-qX;q^2)_{\infty}^{-1}.
$$
For $|q|<1$ the product~\eqref{eq:compact-qdl} is convergent, and for $\Im(b^2)>0$ the compact and non-compact versions of the quantum dilogarithm function are related by
$$
\varphi_b(z) = \frac{\Psi_{\tilde q^{-1}}\big(e^{2\pi b^{-1}z}\big)}{\Psi_{q}\big(e^{2\pi bz}\big)},
\qquad\text{where}\qquad
q = e^{\pi i b^2}, \quad \tilde q = e^{\pi i b^{-2}}.
$$
The analog of the functional equation~\eqref{eqn:func} is given by
\beq
\label{eqn:q-Gamma}
\Psi_q(qX) = (1+X)\Psi_q(q^{-1}X).
\eeq
We will use the following discrete (or ``compact'') analogs of the integral identities from the previous section.
A discrete analog of the Fourier transform formula~\eqref{Fourier-1} is given by the Taylor series expansion of the Pochhammer symbol:
$$
(z;q)^{-1}_{\infty} = \sum_{n\ge0}\frac{z^n}{(q;q)_n},
$$
or more suggestively
$$
\frac{(q;q)_\infty}{(z;q)_{\infty}} = \sum_{n\ge0}{(q^{n+1};q)_\infty}z^n,
$$
with the constant $(q;q)_\infty$ playing the role of phase $\zeta$ in~\eqref{Fourier-1}.
A discrete analog of the pentagon integral~\eqref{beta-1} is given by the $q$-binomial theorem
$$
\sum_{n\ge0}\frac{(a;q)_n}{(q;q)_n}z^n = \frac{(az;q)_{\infty}}{(z;q)_{\infty}}, \quad |z|<1, \quad |q|<1,
$$
or equivalently
$$
\sum_{n\ge0} \frac{(q^{n+1};q)_\infty}{(aq^{n};q)_\infty} = (q;q)_\infty \cdot \frac{(az;q)_{\infty}}{(a;q)_\infty (z;q)_{\infty}}.
$$
%Let us now recall summation formulas, analogous to the integral ones in the previous section. The Ramanujan's $_1\psi_1$ formula can be written as:
%$$
%\sum_{n \in \Z} \frac{(q^nv;q)_\infty}{(q^nu;q)_\infty} w^n = \frac{(v/u;q)_\infty (q/w;q)_\infty}{(v/uw;q)_\infty} \frac{\theta_q(uw)}{\theta_q(u) \theta_q(w)} (q,q)_\infty^2,
%$$
%where
%$$
%\theta_q(x) = \sum_{n \in Z} q^{n(n-1)/2} (-x)^n.
%$$
As for the integral~\eqref{Saal}, in terms of the $q$-hypergeometric series
$$
_2\psi_1(a,b;c;z \,|\, q) = \sum_{n\ge0}\frac{(a;q)_n(b;q)_n}{(c;q)_n(q;q)_n}z^n
%=\frac{\Psi(-q^{-1}c)\Psi(-q)} {\Psi(-q^{-1}a)\Psi(-q^{-1}b)}\sum_{m\ge0}\frac{\Psi(-aq^{2m-1})\Psi(-bq^{2m-1})}{\Psi(-q^{2m+1})\Psi(-cq^{2m-1})}z^m
$$
we have Heine's $q$-analogue of Gauss' summation formula:
\beq
\label{Heine}
_2\psi_1\left(a,b;c;\frac{c}{ab} \,\Big|\, q\right) = \frac{(\frac{c}{a};q)_\infty (\frac{c}{b};q)_\infty}{(c;q)_\infty (\frac{c}{ab};q)_\infty}.
\eeq
Heine's formula can also be written in the following form more reminiscent of~\eqref{Saal}:
$$
\sum_{n\ge 0} \frac{(c q^n,q)_\infty (q^{n+1},q)_\infty}{(aq^n,q)_\infty (bq^n,q)_\infty} z^n = (q,q)_\infty  \frac{(c/a,q)_\infty(c/b,q)_\infty}{(a,q)_\infty (b,q)_\infty (c/(a b),q)_\infty},
$$
In the special case $b = q^{-n}$ with $n\in\mathbb{Z}_{\ge0}$ the right hand side of~\eqref{Heine} terminates and reduces to the Chu--Vandermonde formula
\beq
\label{eq:gauss-thm}
_2\psi_1\left(a,q^{-n};c;q^n\frac{c}{a} \,\Big|\, q\right) = \frac{(c/a;q)_n}{(c;q)_n}.
\eeq
%
%
%
%so that the $q$-binomial coefficient can be expressed as
%$$
%\binom{n}{k}_q = \frac{(q;q)_n}{(q;q)_k(q;q)_{n-k}}.
%$$

\bibliographystyle{alpha}

\end{document}